\documentclass[10pt,journal,compsoc]{IEEEtran}
\ifCLASSOPTIONcompsoc
  \usepackage[nocompress]{cite}
\else
  \usepackage{cite}
\fi
\ifCLASSINFOpdf
   \usepackage[pdftex]{graphicx}
\else
\fi
\usepackage{amsmath}
\usepackage{url}

\newcounter{MYtempeqncnt}

\makeatletter

\newcommand{\Rmnum}[1]{\expandafter\@slowromancap\romannumeral #1@}
\makeatother

\usepackage{float}
\usepackage{epstopdf}
\usepackage{stfloats}

\usepackage{subfigure}
\usepackage{diagbox, eqparbox, hhline}

\usepackage{algorithm}
\usepackage{algpseudocode}
\usepackage{bm}
\usepackage{multirow}
\usepackage{enumerate}
\usepackage{amsthm}
\usepackage[table]{xcolor}
\definecolor{tomato}{RGB}{255,99,71}
\definecolor{db}{RGB}{30,144,255} 

\usepackage[T1]{fontenc}
\usepackage[font=small,labelfont=bf,tableposition=top]{caption}

\DeclareCaptionLabelFormat{andtable}{#1~#2  \&  \tablename~\thetable}

\newtheorem{Theorem}{Theorem}

\algnewcommand{\Initialize}[1]{%
	\State \textbf{initialize:}
	\Statex \hspace*{\algorithmicindent}\parbox[t]{.8\linewidth}{\raggedright #1}
}

\usepackage{tabularx}
\usepackage{threeparttable}

\usepackage[T1]{fontenc}
\usepackage[utf8]{inputenc}

\begin{document}
%
\title{Joint Head Selection and Airtime Allocation for Data Dissemination in Mobile Social Networks}
%
%
%
%

\author{Zhifei Mao,
        Yuming Jiang,
        Xiaoqiang Di,
        and Yordanos Woldeyohannes
\IEEEcompsocitemizethanks{\IEEEcompsocthanksitem Z. Mao, Y. Jiang, and Y. Woldeyohannes are with the Department of Information Security and Communication Technology, NTNU, Norwegian University of Science and Technology, Trondheim N-7491, Norway.\protect\\
E-mail: \{zhifeim,jiang,yordanow\}@item.ntnu.no
\IEEEcompsocthanksitem X. Di is with the Department of Network Engineering, Changchun University of Science and Technology, Changchun 7089, China.\protect\\
E-mail: dixiaoqiang@cust.edu.cn}
}

\IEEEtitleabstractindextext{%
\begin{abstract}
By forming a temporary group, users in mobile social networks (MSNs) can disseminate data to others in proximity with short-range communication technologies. However, due to user mobility, airtime available for users in the same group to disseminate data is limited. In addition, for practical consideration, a star network topology among users in the group is expected. For the former, unfair airtime allocation among the users will undermine their willingness to participate in MSNs. For the latter, a group head is required to connect other users. These two problems have to be properly addressed to enable real implementation and adoption of MSNs. To this aim, we propose a joint head selection and airtime allocation scheme for data dissemination within the group using Nash bargaining theory. Specifically, we consider two cases in terms of user preference on the data to be disseminated: a homogeneous case and a heterogeneous case. For each case, a Nash bargaining solution (NBS) based optimization problem is proposed. The existence of optimal solutions to the optimization problems is proved, which guarantees Pareto optimality and proportional fairness. Next, an algorithm that allows distributed implementation is introduced. Finally, numerical results are presented to evaluate the performance, validate intuitions and derive insights of the proposed scheme. 


\end{abstract}

}

\maketitle

\IEEEdisplaynontitleabstractindextext

%
\IEEEpeerreviewmaketitle

\IEEEraisesectionheading{\section{Introduction}\label{sec:introduction}}
%
%
%
%
%
%

\IEEEPARstart{M}{obile} social networks (MSNs) enable people to share content and communicate without Internet access by exploiting short-range wireless communication technologies such as WiFi Direct and Bluetooth \cite{kayastha2011applications,hu2015survey,mao2017mobile}. Due to the nature of intermittent connectivity, MSNs are often regarded as a special type of delay tolerant network that utilizes opportunistic contacts among mobile users to deliver data \cite{bulut2012exploiting,xiao2014community,jedari2018survey}. Nowadays, people are becoming increasingly inseparable from their portable smart devices such as smartphones. This brings numerous opportunities for people to form temporary groups to exchange information when their portable devices are within each other's transmission range. In particular, MSNs are promising communication systems for people in areas where Internet access is unavailable or too costly. For example, when disasters strike, Internet infrastructure such as cellular networks are among the first pieces of critical infrastructure to fail, leaving individuals disconnected from one another and from vital information sources \cite{nsfmz}. In such scenarios, MSNs will be one of the fastest and most handy ways to provide digital connection among individuals.

Over the past years, significant MSN research effort has been conducted. However, the main focus has been on routing, data dissemination and community detection, leaving a fundamental MSN problem nearly completely untouched, which is local resource management \cite{mao2017mobile}. This surprising phenomenon is probably due to that for various types of wireless networks, local resource management has been extensively studied, and consequently one could expect that their existing solutions might be directly applied or easily extended to MSNs. Unfortunately, this expectation ignores a fundamental difference between MSNs and other popular types of wireless or mobile networks. The difference is that, in the latter cases, there exist base stations (BSs) or access points (APs) for mobile users to get connected to and through the Internet, and in such cases, local resource management, e.g. scheduling the use of airtime among users, typically implicitly assumes that the BSs or APs are always willing to serve, and hence considers only user devices in making the decision. 

However, in MSNs, that seemingly unquestionable assumption may not hold, particularly when a user needs to use her/his smart device to store-carry-forward data for the others  \cite{zhu2015data}. This is because, a smart device normally has limited capacities in terms of e.g. energy, storage, processing and communication. In consequence, local resource management in MSNs not only should consider the devices with data to send or receive, but also must not forget the helpers that contribute additionally in terms of local resources such as energy and communication to act like BSs or APs to help the others. This partly explains why research on MSNs has been progressing for more than one decade but real implementations and adoptions of MSNs in the public are rarely seen today: Overlooking the additional costs incurred to the helper essentially discourages anyone to be helper, which is a foundation for MSNs to work. Nevertheless, the concept of MSNs has drawn huge attention of industry besides academy. Recently, Mozilla and the U.S. National Science Foundation have been running a contest seeking innovative solutions to connect people who are disconnected from the Internet due to disaster or insufficient connectivity \cite{nsfmz}. This provides a great opportunity to work on the missing pieces in MSN research towards public adoption of MSNs, and motives the work of this paper.

In this paper, we focus on a fundamental problem in local resource management, which is \textit{airtime allocation} among users in a temporary group formed on the move, since available airtime is typically limited in MSNs due to user mobility and short transmission range. 
The airtime required for disseminating a piece of data from its source to all the interested group members depends on its communication network topology. In this work, we consider that each group uses a star topology to communicate where one user is selected as \textit{group head} to serve like a personal hot-spot open to the group and manage the group while other users connect to the head as \textit{peripherals}\footnote{Group head and peripherals are called group owner and clients in Wi-Fi Direct terminology \cite{wd}, and called master and slaves in Bluetooth terminology \cite{bt}.}. Such a star form is simple yet practical\footnote{Theoretically, mesh topology is also possible for connecting MSN users in proximity. However, it requires additional functionalities such as multi-hop routing and topology management implemented on users' portable devices. Similarly, while wireless channels are broadcast in nature, using it for applications usually also requires changes to the applications and the various layers below, to be able to make use of this feature.} because it is natively supported by the most popular off-the-shelf short-range communication technologies on portable devices, including WiFi Direct \cite{wd} and Bluetooth \cite{bt}, making the underlying network functionalities transparent to application development. 


With star topology, a group head must be selected among the users. The head needs to forward data for the peripheral users, and thus spends more battery power than them. Therefore, it is important to encourage users to be the head. In addition, since users may have different battery levels and link capacities, \textit{head selection} is critical in that it impacts users' utilities and the amount of data that can be disseminated with the limited airtime. In the literature, various fair airtime (or rate) allocation schemes have been proposed for traditional WLANs and cellular access networks \cite{tan2004time,jiang2005proportional,ramjee2006generalized,li2008proportional}. However, they all implicitly assume that the airtime is long enough so that all the data transmission can be finally completed. This assumption does not generally hold in MSNs where the contact duration among users is limited due to user mobility and short transmission range. In addition, the utility function they use, which is typically $u(x)$ where $x$ is the allocated airtime (or rate), cannot characterize the specifics of users in MSNs, including data dissemination need, preference on other users' data and battery level. Furthermore, unlike previous work, the group head, counterpart of AP in WLANs \cite{li2008proportional} and BS in cellular networks \cite{ramjee2006generalized} respectively,  must also be a target of the airtime allocation, equivalent to the other users connecting to it. These add more difficulty in designing a fair airtime allocation scheme for local data dissemination in MSNs. 

In this paper, we address airtime allocation jointly with head selection among a group of users in an MSN, which, to the best of our knowledge, has never been considered previously. Since anyone in the group may or may not (want to) be the head, a game-theoretic approach is naturally adopted. Specifically, we formulate the problem of joint head selection and airtime allocation as a Nash bargaining game. An advantage of using Nash bargaining is that the solution, if it exists, is known to be Pareto optimal, proportionally fair, and acceptable by all users. Motivated by this, we prove the existence of optimal solution to the joint head selection and airtime allocation Nash bargaining game using decomposition. In addition, we propose a distributed algorithm for joint head selection and airtime allocation, based on the decomposition idea. Moreover, numerical results are presented to provide an overview of the performance, validate intuitions and derive insights of the proposed joint head selection and airtime allocation scheme.

The rest of this paper is organized as follows. In Section~\ref{sec_system}, we introduce the system model including network model, dissemination model, incentive scheme, and user utility function. The Nash bargaining solution (NBS)-based head selection and airtime allocation scheme is proposed and studied in Section \ref{sec_nbs}. In Section~\ref{sec_results}, we show the numerical results. Section \ref{sec_rw} presents related work. Finally, we conclude in Section \ref{sec_con}.  


\section{System Model} 
\label{sec_system} 

Since there are many notations in this paper, we summarize them in Table \ref{table-notations} for reader's convenience. 
\begin{table}[h]
	\caption{Notations.}
	\label{table-notations}
	\centering
	\begin{tabular}{cl}
		\hline        
		Notation	& \multicolumn{1}{c}{Description} \\
		\hline
		$\mathcal{G}$	& User group \\
		$N$	& Number of users in $\mathcal{G}$ \\
		$\mathcal{L}$	& Set of all directed links \\
		$(i,j)$	& Link that sends data from node $i$ to node $j$ \\
		$c_{ij}$	& Rate of link $(i,j)$ \\
		$\mathcal{M}_{i}$	& Set of data that node $i$ intends to disseminate \\
		$\gamma$	& Unit reward \\
		$f_{i}$	& Amount of data node $i$ forwards for other nodes \\
		$u_{i}$	& Utility of node $i$ \\
		$v(\cdot)$	& Valuation function \\
		$g(\cdot)$	& Cost function \\
		$d_{i}$	& Amount of data $i$ disseminates \\
		$b_{i}$	& Amount of data of interests $i$ receives \\
		$e_{i}$	& Total energy consumption of node $i$ \\
		$a_{i}$	& Head indicator \\
		$e^{s}$	& unit energy consumption for sending data \\
		$e^{r}$	& unit energy consumption for receiving data \\
		$s_{i}$	& Amount of data node $i$ sends \\
		$r_{i}$	& Amount of data node $i$ receives \\
		$E_{i}$	& Energy budget of node $i$ \\
		$\delta_{i}$	& Node $i$'s sensitivity to battery power consumption \\
		$x_{i}$	& Airtime for the dissemination of node $i$'s all data \\
		$\alpha_{i}$	& Bargaining power of node $i$ \\
		$T$	& Available airtime \\
		$x_{i}^{m}$	& Airtime for disseminating node $i$'s data $m$ \\
		\hline
	\end{tabular} \\
\end{table}

\subsection{Network Model}

Consider a group\footnote{We assume nearby groups use different channels for data dissemination and data transmission on each channel is independent from the other channels, e.g. in WiFi-Direct. } of users (or nodes) $\mathcal{G} = \{1,2,...,N\}$ in an MSN, which come into contact by opportunity and would like to disseminate their data to other interested nodes in this group. The nodes can communicate with each other by forming a star network $(\mathcal{G,L})$ where $\mathcal{L}$ is the set of all directed links. One of the nodes is selected as the head of the group while other nodes, referred to as peripheral nodes, connect to each other through the head. Denote $c_{ij}$ the rate of link $(i,j)$ that sends data from node $i$ to node $j$. We assume the links in $\mathcal{L}$ may have different rates. 

\subsection{Dissemination Model}

Denote $\mathcal{M}_{i}$ the set of data that node $i$ intends to share to other interested nodes in the group during the contact (There may be some nodes that do not have any data to disseminate but are interested in other nodes' data.). Given that the data of a peripheral node can interest multiple nodes in the group, the head can intentionally store the data (or part of the data) once receiving it from the source node for the first time and then forward it to the rest recipients, so that the limited airtime can be utilized more efficiently than directly sending multiple times from the source node to each recipient. 

\subsection{Incentive Scheme}

Forwarding data for the peripheral nodes will incur a high cost to the energy, storage, etc., therefore, rational nodes are not willing to be the head and forward data for others unconditionally. To encourage nodes to become the head, we assume there is an incentive scheme such that the forwarding behavior is rewarded by the system. Note that the peripheral users do not have to pay to the head for forwarding their data: In practice, such a reward could be in various forms such as popularity and/or reputation in the MSN. For simplicity of analysis, in this paper, we do not restrict the form of implementing the reward and use a linear abstract form of rewarding function, i.e., the node will receive a reward of $\gamma\cdot f$ if it forwards an amount $f$ of data for others, where $\gamma$ is the unit reward. 

\subsection{User Model}

Nodes are effectively autonomous agents, since there is no network-wide control authority. Each node can decide, on its own will, whether to join the group and contribute resources to facilitate data dissemination. In addition, the node selected as the head contributes more resources than client nodes. Therefore, it is reasonable to assume that each node seeks to maximize its utility from data dissemination over a contact. Denote $u_i$ the utility of node $i$, it is given by the valuation of the data it disseminates and the data of interests it receives, minus the energy cost for sending/receiving data, plus the reward for forwarding data for others if $i$ is the head:
\begin{equation}\label{utility}
u_i = v(d_{i}+b_{i})-g(e_{i})+a_{i}\gamma f_{i} 
\end{equation}
where $v(\cdot)$ is the valuation function, $g(\cdot)$ is the cost function, $d_{i}$ is the amount of data $i$ disseminates, $b_{i}$ is the amount of data of interests $i$ receives, $f_{i}$ is the amount of data $i$ forwards for other nodes if it is the head, $e_{i}$ is the total energy consumption for sending and receiving data, and $a_{i}= \{1,0\}$ is the head indicator. For any node $i$, $a_{i} = 1$ means it is selected as the group head while $a_{i} = 0$ means it is a peripheral node. Since there will be only one head, we have $\sum_{i=1}^{N} a_{i} = 1$. Denote $e^{s}$ and $e^{r}$ the unit energy consumption for sending and receiving data, respectively. Then we have $e_{i} = e^{s}s_{i}+e^{r}r_{i}$ where $s_{i}$ is the amount of data it sends and $r_{i}$ is the amount of data it receives. To clarify the difference between $d_{i}$ and $s_{i}$, and the difference between $b_{i}$ and $r_{i}$, an example is illustrated in Fig. \ref{fig-diff}. 
\begin{figure} [!h]
	\centering
	\includegraphics[width=0.32\textwidth]{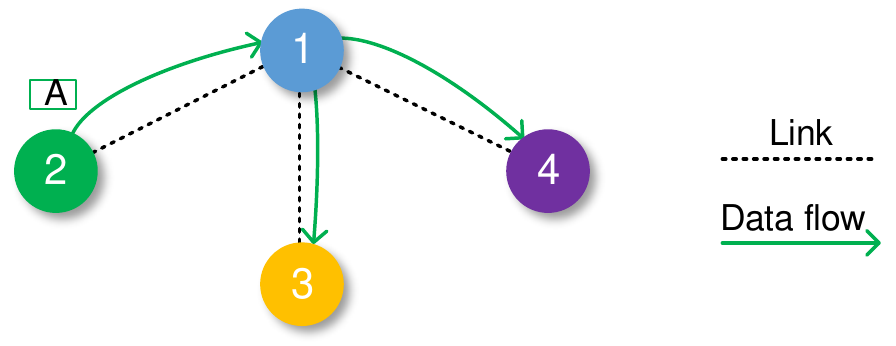}
	\caption{Differences between $d_{i}$ and $s_{i}$, and between $b_{i}$ and $r_{i}$. Node $3$ and $4$ are interested in node $2$'s data \texttt{A} of $5$ MB, but node $1$ is not. Since they cannot communicate directly, node $2$ will send \texttt{A} to the head, i.e., node $1$ first, then node $1$ will forward data \texttt{A} to node $3$ and $4$. As we can see, the amount of data node $2$ disseminates is $d_{2} = 2\times5 = 10$ MB, while the amount of data it directly sends is $s_{2} = 5$ MB. For node $1$, the amount of data it receives is $r_{1} = 5$ MB, while the amount of data of interests it receives is $b_{1} = 0$ MB since it is not interested in data \texttt{A}.}
	\label{fig-diff}
\end{figure}

For the valuation function $v(\cdot)$, we assume it is a strictly concave, positive, and increasing function of $d_{i}+b_{i}$, and $v(\cdot) = 0$ if $d_{i}+b_{i} = 0$. Function $v(d_{i}+b_{i}) = \log(1+d_{i}+b_{i})$ satisfies the above assumptions. Such logarithmic function has been often used in the literature (e.g., \cite{tbma2006game,park2010game,iosifidis2017efficient}) to model a network user's satisfaction or evaluation over certain network resources. For the energy cost function $g(\cdot)$, we assume it is a strictly convex, positive, and increasing function of $e_{i}$, and $g(\cdot) = 0$ if $e_{i} = 0$. In addition, each node $i$ has an energy budget of $E_{i}$ that can be spent during the contact period. Clearly, we need to have $e_{i} \leq E_{i}$.
Function $g(e_{i}) = \delta_{i}(\frac{1}{E_{i}-e_{i}} - \frac{1}{E_{i}})$ satisfies the above assumptions, where $\delta_{i} \in [0,1]$ is a normalization parameter that indicates user $i$'s sensitivity to battery power consumption\footnote{This cost function is a modified version of that used in \cite{iosifidis2017efficient} which does not satisfy $g(0) = 0$.}. For example, a user may have high sensitivity when battery charging is inconvenient. As a rational node will not participate in the group if it will become worse off, it requires $u_i \geq 0$ for all $i \in \mathcal{G}$.


\section{Nash Bargaining to Head Selection And Airtime Allocation} 
\label{sec_nbs}

When a number of users come into each other's proximity, they create a contact opportunity to form a group to exchange data with interested ones. Before they can do that, they have to make a proper decision on head selection and airtime allocation. Since the users are autonomous and rational, each of them would like to benefit from the contact by disseminating its data, receiving data of interests, or obtaining reward. However, the airtime can be very limited due to their mobility so that it would be impossible that everyone gains as much as he/she wants. Therefore, the final decision of head selection and airtime allocation should be acceptable to everyone in order to resolve conflicts of interest. Otherwise, there would be no guarantee that the group will be formed. For such bargaining problems where players not only have incentive to cooperate but also have incentive to oppose each other, Nash bargaining solution (NBS) is an axiomatic approach that can uniquely identify an outcome by its four axioms. In this section, we use Nash bargaining to formulate the problem of airtime allocation jointly with head selection, and analyze the existence of its optimal solution. First of all, we review the basics of NBS in the following section.

\subsection{Basics of Nash Bargaining Solution}\label{sec_nbsbasics}

In this section, we briefly review the concepts and results related to NBS. Consider a bargaining game of $N$ players who bargain or compete for a share of a limited resource (airtime in our case). Throughout the game, the players either reach an agreement on an allocation of the resource or come into disagreement. Let $x_{i}$ be the share of the resource that player $i$ gets, $\mathbf{x} = (x_{1},x_{2},...,x_{N})$ is called a feasible allocation. For each player $i$, it has a utility function $u_{i}(\mathbf{x}): ~ \mathcal{X} \rightarrow \mathcal{R}$ where $\mathcal{X} \subset \mathcal{R}^{N}$ is the set of all possible allocations. Denote $u^{d}_{i}$ the utility of player $i$ when the players come into disagreement, $\mathbf{u^{d}}=(u^{d}_{1},u^{d}_{2},...,u^{d}_{N})$ is called the disagreement point. Then a bargaining game can be formally given by the pair $(\mathcal{U},\mathbf{u^{d}})$ where $\mathcal{U}$ is the set of all feasible utility vectors $\mathbf{u}=(u_{1},u_{2},...,u_{N})$.

Let $ \psi: (\mathcal{U},\mathbf{u^{d}})  \rightarrow \mathcal{R}^{N}$ a bargaining solution that assigns to the bargaining game $(\mathcal{U},\mathbf{u^{d}})$ an element of $\mathcal{U}$. $\psi(\mathcal{U},\mathbf{u^{d}})$ is said to be an NBS if the following axioms are satisfied:
\begin{itemize}
	\item PAR (Pareto optimality). For any $\mathbf{t},\mathbf{t}' \in \mathcal{U}$, if $t_{i} > t'_{i} \text{ for all } i$, then $\psi(\mathcal{U},\mathbf{u^{d}}) \neq \mathbf{t}'$.
	\item ILT (Independence of Linear Transformations). Suppose that the game $(\mathcal{V},\mathbf{v^{d}})$ is obtained from $(\mathcal{U},\mathbf{u^{d}})$ by the transformations $v_{i}=\sigma_{i}u_{i}+\theta_{i}$, $\sigma_{i} > 0$ for all $i$, then  $\psi_{i}(\mathcal{V},\mathbf{v^{d}}) = \sigma_{i}\psi_{i}(\mathcal{U},\mathbf{u^{d}})+\theta_{i}$ for all $i$.
	\item SYM (Symmetry). If $\mathcal{U}$ is invariant under the exchanges of player $i$ and player $j$ and $u^{d}_{i}=u^{d}_{j}$, then $\psi_{i}(\mathcal{U},\mathbf{u^{d}})=\psi_{j}(\mathcal{U},\mathbf{u^{d}})$, for all possible $i,j$.
	\item IIA (Independence of Irrelevant Alternatives). If $(\mathcal{U},\mathbf{u^{d}})$ and $(\mathcal{V},\mathbf{u^{d}})$ are two bargaining games with $\mathcal{V} \subset \mathcal{U}$ and $\psi(\mathcal{U},\mathbf{u^{d}}) \in \mathcal{V}$, then $\psi(\mathcal{U},\mathbf{u^{d}}) = \psi(\mathcal{V},\mathbf{u^{d}})$.
\end{itemize}
PAR ensures no wastage in the resource. ILT states that the bargaining solution is invariant with respect to linear utility transformations. SYM means that if any two players have the same utility function and disagreement utility, they will have the same utility in the bargaining solution. IIA says that if the feasible utility
set shrinks, but the bargaining solution remains feasible in the smaller set, then the bargaining solution to the game with the smaller utility set should be the same. The latter three axioms (i.e., ILT, SYM and IIA) are often regarded as axioms of fairness \cite{yaiche2000game,park2007bargaining}, as they allow NBS to select a fair allocation among the set of all Pareto optimal allocations. More details and interpretations of NBS can be found in \cite{hargreaves2004game}.

Assuming the utility set $\mathcal{U}$ is compact convex and there is at least one $\mathbf{u}$ such that $u_{i} > u^{d}_{i}$ for all $i$, then there exists a unique bargaining solution fulfilling the above four axioms, which maximizes the following Nash product (or Nash welfare) \cite{yaiche2000game}:
\begin{equation}\label{eq-np}
\textstyle\prod\limits^{N}_{i=1} (u_{i}(\mathbf{x})-u^{d}_{i}).
\end{equation}
Though no explicit fairness is defined within the four axioms, NBS shows strong fairness property. 
It is well-known that when $u^{d}_{i} = 0$ for all $i$, NBS guarantees proportional fairness (PF) in utility. An allocation that satisfies PF should be that, moving away from the PF allocation or NBS to any other feasible allocation will not increase the aggregate of proportional changes in utilities \cite{han2005fair,shrimali2010cooperative,bertsimas2011price}. In mathematical terms, $\sum^{N}_{i=1} \frac{u_{i}-u_{i}^{\star}}{u_{i}^{\star}} \leq 0$
where $\mathbf{u^{\star}}=(u_{1}^{\star},u_{2}^{\star},...,u_{N}^{\star})$ is the PF allocation and $\mathbf{u}=(u_{1},u_{2},...,u_{N})$ is any other feasible allocation. Due to such relationship, NBS is often regarded as a generalization of proportional fairness. By relaxing the axiom of SYM \cite{binmore1986nash,yaiche2000game}, the so-called generalized (or asymmetric) NBS can be obtained by maximizing
\begin{equation}\label{eq-gnp}
\textstyle\prod\limits^{N}_{i=1} (u_{i}(\mathbf{x})-u^{d}_{i})^{\alpha_{i}} 
\end{equation}
where $0\leq\alpha_{i}\leq 1$ is the bargaining power and $\sum_{i=1}^{N} \alpha_{i} = 1$. Generalized NBS satisfies the axioms of PAR, ILT and IIA and guarantees weighted proportional fairness which satisfies $\sum^{N}_{i=1} \alpha_{i}\cdot\frac{u_{i}-u_{i}^{\star}}{u_{i}^{\star}} \leq 0$ \cite{cole2013mechanism}.

In the following, we will first elaborate the utility function of users in the cases of homogeneous user preference and heterogeneous user preference, model the head selection and airtime allocation using generalized NBS, and then discuss the existence of optimal solution for both cases. The intention of using generalized NBS instead of standard NBS is to see whether bargaining power allows the head to be selected to gain higher utility than other users, which motivates the users to become the head willingly.

\subsection{Homogeneous User Preference}

\begin{figure*}[!b]
	\hrulefill
	\normalsize
	\setcounter{MYtempeqncnt}{\value{equation}}
	\setcounter{equation}{14}
	\begin{align} \label{exact-utility}
	u_i = & v\Big(\frac{N_{i} x_{i}}{\sum_{(k,j)\in \mathcal{L}_{i}} \frac{1}{c_{kj}}} + \sum_{h\in \mathcal{G}_{-i}} \frac{x_{h}}{\sum_{(k,j)\in \mathcal{L}_{h}}\frac{1}{c_{kj}}}\Big) + a_{i}\gamma \sum\limits_{h\in \mathcal{G}_{-i}} (N_{h}-\beta_{h}) \frac{x_{h}}{\sum_{(k,j)\in \mathcal{L}_{h}}\frac{1}{c_{kj}}} \notag\\
	& -  g\Big(e^{r}\sum_{h\in \mathcal{G}_{-i}} \frac{x_{h}}{\sum_{(k,j)\in \mathcal{L}_{h}}\frac{1}{c_{kj}}} + e^{s}(1-a_{i})\frac{\beta_{i}x_{i}}{\sum_{(k,j)\in \mathcal{L}_{i}} \frac{1}{c_{kj}}} + e^{s}a_{i}\big(\frac{N_{i} x_{i}}{\sum_{(k,j)\in \mathcal{L}_{i}} \frac{1}{c_{kj}}} + \sum\limits_{h\in \mathcal{G}_{-i}}  \frac{(N_{h}-\beta_{h}) x_{h}}{\sum_{(k,j)\in \mathcal{L}_{h}}\frac{1}{c_{kj}}}\big)\Big)
	\end{align}
	\setcounter{equation}{\value{MYtempeqncnt}}
\end{figure*}

Assume the nodes have homogeneous preference on the data, i.e., they are interested in any data that any other nodes would like to disseminate. Define a \textit{dissemination} of the data of any node $i$ the set of transmissions (or links) that send $i$'s data from one node to the other. For a peripheral node, its dissemination includes the transmission from itself to the head and $N-2$ transmissions from the head to other nodes in the group. For the head, its dissemination consists of $N-1$ transmissions from itself to all the peripheral nodes. Denote $x_{i}$ the airtime for the dissemination of node $i$'s data, then the airtime constraint is given by
\begin{equation}\label{cons-time-nopref}
\sum_{i=1}^{N} x_{i} \leq T
\end{equation}
where $T$ is the available airtime. For each node $i$, we have 
\begin{equation}\label{cons-data-nopref}
0 \leq x_{i} \leq \sum_{(k,j)\in \mathcal{L}_{i}}\frac{z_{i}}{c_{kj}}
\end{equation}
where $z_{i}$ is the size of all the data in $\mathcal{M}_{i}$ and $\mathcal{L}_{i}$ is the set of all the links that disseminate $i$'s data. The constraint shown in (\ref{cons-data-nopref}) means that the airtime allocated to $i$'s dissemination should not exceed what it needs.

Within the dissemination of any node's data, we also aim a fair data distribution among all the transmissions. Ideally, the progress of all the transmissions of a given dissemination, defined as the amount of data transmitted, should be equal when the dissemination stops. 
Mathematically, we have
\begin{equation}\label{xlink1-nopref}
\theta_{i} = c_{kj} x_{kj}, \forall (k,j)\in \mathcal{L}_{i}
\end{equation}
where $\theta_{i}$ denotes the amount of data transmitted by every link in $\mathcal{L}_{i}$ and $x_{kj}$ is the airtime for link $(k,j)$ in $\mathcal{L}_{i}$ to disseminate node $i$'s data. Now we can express the airtime for sending $i$'s data via each link in $\mathcal{L}_{i}$ in terms of $\theta_{i}$:
\begin{equation}\label{xlink2-nopref}
x_{kj}  = \frac{\theta_{i}}{c_{kj}}, \forall (k,j)\in \mathcal{L}_{i}.
\end{equation}
Since
\begin{equation}\label{xlink3-nopref}
\sum_{(k,j)\in \mathcal{L}_{i}} x_{kj}  = \sum_{(k,j)\in \mathcal{L}_{i}} \frac{\theta_{i}}{c_{kj}} = x_{i},
\end{equation}
the amount of $i$'s data transmitted by each link $(k,j)\in \mathcal{L}_{i}$ can be given by
\begin{equation}\label{xi-nopref}
\theta_{i} = \frac{x_{i}}{\sum_{(k,j)\in \mathcal{L}_{i}} \frac{1}{c_{kj}}}.
\end{equation}
Then, the amount of data $i$ disseminates within $T$ is given by
\begin{equation}\label{di-nopref}
d_{i} = \sum_{(k,j)\in \mathcal{L}_{i}}\theta_{i} = \frac{N_{i} x_{i}}{\sum_{(k,j)\in \mathcal{L}_{i}} \frac{1}{c_{kj}}} 
\end{equation}
where $N_{i}$ is the number of transmissions in the dissemination of $i$'s data in $T$. $N_{i} = N-1$ if the head has not stored $i$'s data, and $N_{i} = N-2$ if the head has. The amount of data of interests $i$ receives within $T$ can be given by 
\begin{equation}\label{bi-nopref}
b_{i} = \sum_{h\in \mathcal{G}_{-i}} \theta_{h} = \sum_{h\in \mathcal{G}_{-i}} \frac{x_{h}}{\sum_{(k,j)\in \mathcal{L}_{h}}\frac{1}{c_{kj}}}.
\end{equation}
where $\mathcal{G}_{-i} = \mathcal{G}\setminus\{i\}$ is the set of users in $\mathcal{G}$ except $i$. If $i$ will be selected as the head, the amount of data $i$ forwards for other nodes is 
\begin{equation} \label{fi-nopref}
f_{i} = \sum\limits_{h\in \mathcal{G}_{-i}} (N_{h}-\beta_{h}) \frac{x_{h}}{\sum_{(k,j)\in \mathcal{L}_{h}}\frac{1}{c_{kj}}}
\end{equation}
where $\beta_{h} = 1$ means the head has not stored node $h$'s data and $h$ needs to send the data to the head, otherwise $\beta_{h} = 0$. For a peripheral node $i$, it only sends its data to the head, therefore $s_{i} = \frac{\beta_{i}x_{i}}{\sum_{(k,j)\in \mathcal{L}_{i}} \frac{1}{c_{kj}}}$. However, for a head $i$, it not only sends its own data but also others' data to all the interested nodes, therefore we have $s_{i} = d_{i} + f_{i}$. Using a unified expression, we have
\begin{align} \label{si-nopref}
s_{i} = & a_{i}\Big(\frac{N_{i} x_{i}}{\sum_{(k,j)\in \mathcal{L}_{i}} \frac{1}{c_{kj}}} + \sum\limits_{h\in \mathcal{G}_{-i}}  \frac{(N_{h}-\beta_{h}) x_{h}}{\sum_{(k,j)\in \mathcal{L}_{h}}\frac{1}{c_{kj}}}\Big) \notag\\
& + (1-a_{i})\frac{\beta_{i}x_{i}}{\sum_{(k,j)\in \mathcal{L}_{i}} \frac{1}{c_{kj}}}.
\end{align}
In the case of homogeneous preference, we have $r_{i} = b_{i}$, i.e., the amount of data \textit{of interests} $i$ receives equals the amount of data $i$ receives. Finally, the total consumed energy of $i$ is
\begin{align} \label{ei-nopref}
e_{i} & = e^{s}(1-a_{i})\frac{\beta_{i}x_{i}}{\sum_{(k,j)\in \mathcal{L}_{i}} \frac{1}{c_{kj}}} + e^{r}\sum_{h\in \mathcal{G}_{-i}} \frac{x_{h}}{\sum_{(k,j)\in \mathcal{L}_{h}}\frac{1}{c_{kj}}} \notag\\
& + e^{s}a_{i}\Big(\frac{N_{i} x_{i}}{\sum_{(k,j)\in \mathcal{L}_{i}} \frac{1}{c_{kj}}} + \sum\limits_{h\in \mathcal{G}_{-i}}  \frac{(N_{h}-\beta_{h}) x_{h}}{\sum_{(k,j)\in \mathcal{L}_{h}}\frac{1}{c_{kj}}}\Big).
\end{align}

Replacing $d_{i}$, $b_{i}$, $f_{i}$, and $e_{i}$ in (\ref{utility}) by (\ref{di-nopref}), (\ref{bi-nopref}), (\ref{fi-nopref}) and (\ref{ei-nopref}), we obtain the utility of any user $i$ shown in (\ref{exact-utility}) at the bottom of this page.
\addtocounter{equation}{1}
\begin{figure*}[!b]
	\hrulefill
	\normalsize
		\setcounter{MYtempeqncnt}{\value{equation}}
		\setcounter{equation}{30}
	\begin{align} \label{exact-utility-hetero}
	u_i = 
	& v\Big(\sum_{m=1}^{M_{i}} \frac{N^{m}_{i} x^{m}_{i}}{\sum_{(k,j)\in \mathcal{L}^{m}_{i}}\frac{1}{c_{kj}}} + \sum_{h\in \mathcal{B}_{i}}\sum_{m\in \mathcal{M}^{h}_{i}} \frac{x^{m}_{h}}{\sum_{(k,j)\in \mathcal{L}^{m}_{h}}\frac{1}{c_{kj}}}\Big) + a_{i}\gamma \sum\limits_{h\in \mathcal{G}_{-i}} \sum_{m=1}^{M_{h}} (N^{m}_{h}-\beta^{m}_{h}) \frac{x^{m}_{h}}{\sum_{(k,j)\in \mathcal{L}^{m}_{h}}\frac{1}{c_{kj}}} \notag\\
	& -g\Bigg((1-a_{i})\Big(e^{s}\sum_{m=1}^{M_{i}} \frac{x^{m}_{i}\beta^{m}_{i}}{\sum_{(k,j)\in \mathcal{L}^{m}_{i}}\frac{1}{c_{kj}}} + e^{r}\sum_{h\in \mathcal{B}_{i}}\sum_{m\in \mathcal{M}^{h}_{i}} \frac{x^{m}_{h}}{\sum_{(k,j)\in \mathcal{L}^{m}_{h}}\frac{1}{c_{kj}}}\Big) + a_{i}\Big(e^{s}\Big(\sum_{m=1}^{M_{i}} N^{m}_{i} \frac{x^{m}_{i}}{\sum_{(k,j)\in \mathcal{L}^{m}_{i}}\frac{1}{c_{kj}}} \\
	& + \sum\limits_{h\in \mathcal{G}_{-i}} \sum_{m=1}^{M_{h}} (N^{m}_{h}-\beta^{m}_{h}) \frac{x^{m}_{h}}{\sum_{(k,j)\in \mathcal{L}^{m}_{h}}\frac{1}{c_{kj}}}\Big) + e^{r}\sum\limits_{h\in \mathcal{G}_{-i}} \sum_{m=1}^{M_{h}} \frac{x^{m}_{h}\beta^{m}_{h}}{\sum_{(k,j)\in \mathcal{L}^{m}_{h}}\frac{1}{c_{kj}}}\Big)\Bigg) \notag
	\end{align}
		\setcounter{equation}{\value{MYtempeqncnt}}
\end{figure*}
From (\ref{exact-utility}), we can see that $u_{i}$ is a function of $\mathbf{x}$ and $\mathbf{a}$ where $\mathbf{x} = (x_{1}, x_{2},...,x_{N})$ and $\mathbf{a} = (a_{1}, a_{2},...,a_{N})$. We assume there is at least one $(\mathbf{x,a})$ makes $u_{i} \geq 0$ $\forall i \in \mathcal{G}$ otherwise nodes have no motivation to join the group and disseminate their data. Formally, the generalized NBS for the problem of head selection and airtime allocation for the case of homogeneous user preference can be obtained by maximizing the following generalized Nash product:
\begin{align}\label{full-nopref}
 \underset{\mathbf{x,a}}{\max}\quad\quad & \textstyle\prod\limits_{i=1}^{N} u_{i}(\mathbf{x,a})^{\alpha_{i}} \\
 \text{s.t.}\quad\quad             & \sum_{i=1}^{N} x_{i} \leq T  \label{c-at}\\
 & 0 \leq x_{i} \leq \sum_{(k,j)\in \mathcal{L}_{i}}\frac{z_{i}}{c_{kj}}, & ~\forall i \in \mathcal{G}  \label{c-atr}\\
 & u_{i}(\mathbf{x,a}) \geq 0, &  ~\forall i \in \mathcal{G}  \label{c-ui}\\
 & e_{i} \leq E_{i}, &  ~\forall i \in \mathcal{G}  \label{c-ei}\\
 & a_{i} = \{1,0\}, &  ~\forall i \in \mathcal{G}  \label{c-head1}\\
 & \sum_{i=1}^{N} a_{i} = 1. \label{c-head2}
\end{align}
Under the generalized NBS framework, a higher generalized Nash product means a better decision of head selection and airtime allocation. In our case, users' utilities are zero at the disagreement point since they will get nothing if no group is formed. Eq. (\ref{c-at}) represents the airtime constraint for all the group members. Eq. (\ref{c-atr}) states that the airtime allocated to the dissemination of any user's data should be nonnegative and not be longer than the maximum airtime required. Eq. (\ref{c-ui}) ensures individual rationality. Eq. (\ref{c-ei}) limits the energy consumption of each user to its energy budget. Finally, Eq. (\ref{c-head1}) and (\ref{c-head2}) indicate that only one of the users would be selected as the group head. 

The problem  (\ref{full-nopref}) has at least one optimal solution. The proof of this statement is skipped, because in Section \ref{sec-hetero} below, a more general case, the heterogeneous case, is studied. For this more general case, it will be proved with details that the same statement holds, as shown in Theorem \ref{theorem-pref}, for the more generalized problem (\ref{full-pref}) that corresponds to the problem  (\ref{full-nopref}) here.

\subsection{Extension To Heterogeneous User Preference}
\label{sec-hetero}

The above model applies to MSN systems where users are interested in the same data. However, in some MSN systems (e.g., publish-subscribe systems), users may be interested in different data. In this section, we extend the above model to cases with heterogeneous user preferences.

Consider that there could be multiple data in $\mathcal{M}_{i}$. Let $\mathcal{L}^{m}_{i}$ be the set of links that disseminate node $i$'s data $m$. Denote $x^{m}_{i}$ the airtime for disseminating $i$'s data $m$. Then the total airtime $x_{i}$ for disseminating $i$'s data is $\sum_{m=1}^{M_{i}} x^{m}_{i}$ where $M_{i} = |\mathcal{M}_{i}|$ is the number of data in $\mathcal{M}_{i}$. Then the airtime constraint is given by
\begin{equation}\label{multirate-preference-cons-time}
\sum_{i=1}^{N}\sum_{m=1}^{M_{i}} x^{m}_{i} \leq T.
\end{equation}
For each node $i$  and its data $m$, we have 
\begin{equation}\label{multirate-preference-cons-data}
0 \leq x^{m}_{i} \leq \sum_{(k,j)\in \mathcal{L}^{m}_{i}}\frac{z^{m}_{i}}{c_{kj}}
\end{equation}
where $z^{m}_{i}$ is the size of data $m$. The total amount of data $i$ disseminates is
\begin{equation}\label{di-pref}
d_{i} = \sum_{m=1}^{M_{i}} N^{m}_{i} \frac{x^{m}_{i}}{\sum_{(k,j)\in \mathcal{L}^{m}_{i}}\frac{1}{c_{kj}}}
\end{equation}
where $N^{m}_{i}$ is the number of transmissions in the dissemination of $i$'s data $m$. Denote $\mathcal{B}_{i}$ the set of nodes that disseminate data to $i$ (or equivalently the set of nodes that $i$ is interested in their data). The amount of data of interests $i$ receives is 
\begin{equation}\label{bi-pref}
b_{i} = \sum_{h\in \mathcal{B}_{i}}\sum_{m\in \mathcal{M}^{h}_{i}} \frac{x^{m}_{h}}{\sum_{(k,j)\in \mathcal{L}^{m}_{h}}\frac{1}{c_{kj}}}
\end{equation}
where $\mathcal{M}^{h}_{i}$ is the set of $h$'s data sent to $i$. If $i$ will be the head after selection, the amount of data it forwards is
\begin{equation}\label{fi-pref}
f_{i}= \sum\limits_{h\in \mathcal{G}_{-i}} \sum_{m=1}^{M_{h}} (N^{m}_{h}-\beta^{m}_{h}) \frac{x^{m}_{h}}{\sum_{(k,j)\in \mathcal{L}^{m}_{h}}\frac{1}{c_{kj}}}
\end{equation}
where $\beta^{m}_{h} = 1$ means the head has not stored $h$'s data $m$ and $h$ needs to send $m$ to the head, otherwise $\beta^{m}_{h} = 0$. The amount of data $i$ sends is given by
\begin{align} \label{si-pref}
s_{i} = & a_{i}\Big(\sum_{m=1}^{M_{i}} N^{m}_{i} \frac{x^{m}_{i}}{\sum_{(k,j)\in \mathcal{L}^{m}_{i}}\frac{1}{c_{kj}}} \notag\\
& + \sum\limits_{h\in \mathcal{G}_{-i}} \sum_{m=1}^{M_{h}} (N^{m}_{h}-\beta^{m}_{h}) \frac{x^{m}_{h}}{\sum_{(k,j)\in \mathcal{L}^{m}_{h}}\frac{1}{c_{kj}}}\Big) \\
& + (1-a_{i})\sum_{m=1}^{M_{i}} \frac{x^{m}_{i}\beta^{m}_{i}}{\sum_{(k,j)\in \mathcal{L}^{m}_{i}}\frac{1}{c_{kj}}}. \notag
\end{align}
For a peripheral node $i$, we still have $r_{i} = b_{i}$. However, for the head, it does not hold, since it may receive some data of no interest and only for forwarding. Then the amount of data $i$ receives is
\begin{align} \label{ri-pref}
r_{i} = & a_{i}\sum\limits_{h\in \mathcal{G}_{-i}} \sum_{m=1}^{M_{h}} \frac{x^{m}_{h}\beta^{m}_{h}}{\sum_{(k,j)\in \mathcal{L}^{m}_{h}}\frac{1}{c_{kj}}} \notag\\
& + (1 - a_{i})\sum_{h\in \mathcal{B}_{i}}\sum_{m\in \mathcal{M}^{h}_{i}} \frac{x^{m}_{h}}{\sum_{(k,j)\in \mathcal{L}^{m}_{h}}\frac{1}{c_{kj}}}. 
\end{align}
Finally, the total consumed energy of $i$ is
\begin{align} \label{ei-pref}
e_{i} = & (1-a_{i})\Big(e^{s}\sum_{m=1}^{M_{i}} \frac{x^{m}_{i}\beta^{m}_{i}}{\sum_{(k,j)\in \mathcal{L}^{m}_{i}}\frac{1}{c_{kj}}}  \notag\\
& + e^{r}\sum_{h\in \mathcal{B}_{i}}\sum_{m\in \mathcal{M}^{h}_{i}} \frac{x^{m}_{h}}{\sum_{(k,j)\in \mathcal{L}^{m}_{h}}\frac{1}{c_{kj}}}\Big) \notag\\
& + a_{i}\Bigg(e^{s}\Big(\sum_{m=1}^{M_{i}} N^{m}_{i} \frac{x^{m}_{i}}{\sum_{(k,j)\in \mathcal{L}^{m}_{i}}\frac{1}{c_{kj}}} \\
& + \sum\limits_{h\in \mathcal{G}_{-i}} \sum_{m=1}^{M_{h}} (N^{m}_{h}-\beta^{m}_{h}) \frac{x^{m}_{h}}{\sum_{(k,j)\in \mathcal{L}^{m}_{h}}\frac{1}{c_{kj}}}\Big)  \notag\\
& + e^{r}\sum\limits_{h\in \mathcal{G}_{-i}} \sum_{m=1}^{M_{h}} \frac{x^{m}_{h}\beta^{m}_{h}}{\sum_{(k,j)\in \mathcal{L}^{m}_{h}}\frac{1}{c_{kj}}}\Bigg).  \notag
\end{align}
Replacing $d_{i}$, $b_{i}$, $f_{i}$, and $e_{i}$ in (\ref{utility}) by (\ref{di-pref}), (\ref{bi-pref}), (\ref{fi-pref}), and (\ref{ei-pref}), we obtain the utility of any user $i$ shown in (\ref{exact-utility-hetero}) at the bottom of this page.
\addtocounter{equation}{1}
Formally, the generalized NBS for the problem of head selection and airtime allocation for the case of heterogeneous user preference can be obtained by maximizing the following optimization problem
\begin{align}\label{full-pref}
 \underset{\mathbf{x,a}}{\max}\quad\quad & \textstyle\prod\limits_{i=1}^{N} u_{i}(\mathbf{x,a})^{\alpha_{i}} \\
 \text{s.t.}\quad\quad             & \sum_{i=1}^{N}\sum_{m=1}^{M_{i}} x^{m}_{i} \leq T \label{c2-at}\\
 & 0 \leq x^{m}_{i} \leq \sum_{(k,j)\in \mathcal{L}^{m}_{i}}\frac{z^{m}_{i}}{c_{kj}}, & ~\forall i \in \mathcal{G}, m \in \mathcal{M}_{i} \\
 & u_{i}(\mathbf{x,a}) \geq 0, &  ~\forall i \in \mathcal{G} \\
 & e_{i} \leq E_{i}, &  ~\forall i \in \mathcal{G} \\
 & a_{i} = \{1,0\}, &  ~\forall i \in \mathcal{G} \\
 & \sum_{i=1}^{N} a_{i} = 1. \label{c2-head2}
\end{align}
Constraints (\ref{c2-at}) to (\ref{c2-head2}) have the same meaning with constraints (\ref{c-at}) to (\ref{c-head2}), respectively. Assuming there is at least one $(\mathbf{x,a})$ makes $u_{i} \geq 0$ $\forall i \in \mathcal{G}$, we have the following theorem.

\begin{Theorem}\label{theorem-pref}
	There exists at least one optimal solution to optimization problem (\ref{full-pref}) for joint head selection and airtime allocation.
\end{Theorem}
\begin{proof}
	In fact, the optimization problem (\ref{full-pref}) has two levels of optimization. At the lower level, each user $i$ in the group solves a sub-problem (a local generalized NBS problem) that finds optimal airtime allocation among all the users when user $i$ is the head. At the higher level, we have a master problem that chooses the best $i$ to be the head, which gives the highest generalized Nash product.
	
	Mathematically, the sub-problem for each user is given by
	\begin{align}\label{lower-nopref}
	 \underset{\mathbf{x}}{\max}\quad\quad    & \textstyle\prod\limits_{i=1}^{N} u_{i}(\mathbf{x})^{\alpha_{i}} \notag\\
	 \text{s.t.}\quad\quad	& \sum_{i=1}^{N}\sum_{m=1}^{M_{i}} x^{m}_{i} \leq T \notag\\
	& 0 \leq x^{m}_{i} \leq \sum_{(k,j)\in \mathcal{L}^{m}_{i}}\frac{z^{m}_{i}}{c_{kj}}, & ~\forall i \in \mathcal{G}, m \in \mathcal{M}_{i} \\
	& u_{i}(\mathbf{x}) \geq 0, &  ~\forall i \in \mathcal{G} \notag\\
	& e_{i} \leq E_{i}, &  ~\forall i \in \mathcal{G}. \notag
	\end{align}
	where $\mathbf{a}$ is fixed and $u_{i}$ is only a function of $\mathbf{x}$. Since $v(\cdot)$ and $-g(\cdot)$ are strictly concave functions, by the concavity preserving rules in \cite{boyd2004convex}, we can see that $u_i$ is a strictly concave function in $\mathbf{x}$. Since the function of $log$ is concave and monotonic, the objective function of problem (\ref{lower-nopref}) is equivalent to \cite{yaiche2000game}
	\begin{equation}\label{equl-lower-nopref}
	\max\limits_{\mathbf{x}} \textstyle\sum\limits_{i=1}^{N} \alpha_{i}\log u_{i}(\mathbf{x})
	\end{equation}
	It is easy to see that (\ref{equl-lower-nopref}) is strictly concave and the constraints in (\ref{lower-nopref}) are convex. Additionally, we have assumed that there is at least one feasible point, meaning the constraint set is non-empty, therefore there exists a unique optimal solution to problem (\ref{equl-lower-nopref}) and equivalently to problem (\ref{lower-nopref}) \cite{boyd2004convex}.
	
	At the higher level, the master problem is 
	\begin{align}\label{master-nopref}
	 \underset{\mathbf{a}}{\max}\quad\quad\quad\quad	& p^{\star}(\mathbf{a}) \notag\\
	 \text{s.t.}\quad\quad\quad\quad	&  a_{i} = \{1,0\}, &  ~\forall i \in \mathcal{G} \\
	 & \sum_{i=1}^{N} a_{i} = 1 \notag
	\end{align}
	where $p^{\star}(\mathbf{a}) = \max\limits_{\mathbf{x}}\prod_{i=1}^{N} u_{i}(\mathbf{x})^{\alpha_{i}}$ is the optimal objective value of problem (\ref{lower-nopref}) for a given $\mathbf{a}$ (namely, a given user being the head). Since there will be only one $a_{i}$ equals $1$ and the rest are $0$, the objective of the master problem (\ref{master-nopref}) is essentially finding the largest within $N$ real numbers, which always exists. The pair(s) $(\mathbf{x,a})$ resulting in the largest number will be the optimal solution(s)\footnote{Strictly speaking, there might be multiple maximum in $N$ real numbers. Therefore, we do not claim uniqueness of the optimal solution.} to the whole problem. 
\end{proof}

\textbf{Remark 1}: In the above proof, we show that there is a unique optimal solution to each local generalized NBS problem (\ref{lower-nopref}), meaning the axioms of PAR, ILT and IIA presented in Section \ref{sec_nbsbasics} are satisfied in maximizing each local generalized NBS problem. Note that the master problem is just finding the user giving the highest generalized Nash product among all users. Suppose user $i$ is finally selected as the head, then the final optimal airtime allocation will be the optimal airtime allocation to the local generalized NBS problem of user $i$, which satisfies the axioms of PAR, ILT and IIA. Therefore, no matter which user is selected as the head, the axioms of PAR, ILT and IIA are always satisfied.

\textbf{Remark 2}: Letting $m = 1$ and $\mathcal{B}_{i} = \mathcal{G}_{-i}$, we can see that it reduces to the model for the homogeneous case. In other words, the homogeneous case is a special instance of the heterogeneous case. Therefore, Theorem \ref{theorem-pref} holds for the homogeneous case as well.

\subsection{Algorithm} 

Based on the idea of decomposition in the proof of Theorem (\ref{theorem-pref}), we present an algorithm that can find a unique optimal solution in a distributed fashion. First of all, each node $i$ solves problem (\ref{lower-nopref}) with $a_{i} = 1$. Since problem (\ref{lower-nopref}) is a convex optimization problem with inequality constraints, its optimal solution can be found by interior point methods\cite{boyd2004convex}. After solving the problem, node $i$ sends the optimal results $(\mathbf{x^{\star}})_{i}$ and $(p^{\star})_{i}$, i.e., the optimal airtime allocation and generalized Nash product given $i$ is the head, to all other nodes in $\mathcal{G}$. Once receiving the optimal results from all other nodes, each node $i$ checks which node being the head will result in the largest generalized Nash product. If node $i$ happens to have the largest generalized Nash product, it will become the group head and $(\mathbf{x^{\star}})_{i}$ will be the final optimal airtime allocation. It is possible that multiple nodes have the largest generalized Nash product and any such node can be the head. In such cases, the node with the lowest index will be selected as the head without loss of generality\footnote{Certainly, there are other approaches to determine the final group head, such as selecting the node with the highest battery power. Actually, in the numerical results, it will be shown that user with high energy budget is preferred to be the head if we do not specifically select them.}. The above steps is summarized in Algorithm \ref{alg-allocation}.
\begin{algorithm}
	\caption{Joint Head Selection And Airtime Allocation}
	\label{alg-allocation}
	\renewcommand{\algorithmicensure}{\textbf{Output:}}
	\begin{algorithmic}[1]
		\Ensure $(\mathbf{x^{\star},a^{\star}})$
		\For{$i \in \mathcal{G}$}
		\State Solve problem (\ref{lower-nopref}) with $a_{i} = 1$ to get $(\mathbf{x^{\star}})_{i}$ and $(p^{\star})_{i}$
		\State Send $(p^{\star})_{i}$ to all other nodes in $\mathcal{G}$
		\While {receiving $(p^{\star})_{k}$ from all $k\neq i \in \mathcal{G}$}
		\If {$i = \min \{\underset{k \in \mathcal{G}}{\arg\max} (p^{\star})_{k}\}$}
		\State  set $a_{i} =1$ and $\mathbf{x^{\star}} = (\mathbf{x^{\star}})_{i}$
		\Else 
		\State  set $a_{i} =0$
		\EndIf
		\EndWhile
		\EndFor
	\end{algorithmic}
\end{algorithm}

\subsection{Handling Dynamics} 

Once the group members are known to all, Algorithm \ref{alg-allocation} will be applied to select a proper head and allocate airtime to each user. In practice, node mobility would introduce dynamics to a formed group, such as node leaving and joining. To cope with such dynamics, we presume a group management function is implemented at each node, which allows the head on duty to update group information periodically. Once an event (e.g., a new node joins) is detected, the head on duty will inform other users to start a new round of head selection and airtime allocation by applying Algorithm \ref{alg-allocation}.

\section{Numerical Results} 
\label{sec_results}

In this section, we demonstrate how the NBS-based approach to the joint head selection and airtime allocation problem performs. Particularly, we show the impacts of different parameters such as energy budget, unit reward and link capacity on the behavior of the proposed approach. The tools used in solving the optimization problems and obtaining the results are AMPL\footnote{A powerful tool that can solve high-complexity optimization problems. (www.ampl.com)} jointly with MATLAB. In the utility function of users, we use $v(d_{i}+b_{i}) = \log(1+d_{i}+b_{i})$ as the valuation function and $g(e_{i}) = \delta_{i}(\frac{1}{E_{i}-e_{i}} - \frac{1}{E_{i}})$ as the cost function.

\begin{table*}[t]
	\caption{Utilities of the users with different users being the head when user $1$'s budget is $300$ Joules.}
	\label{table-eb-ut}
	\centering
	\begin{threeparttable}
		\begin{tabular}{c|cccc|cccc}
			\hline
			Sensitivity      & \multicolumn{4}{c|}{$\mathbf{\delta} = [0,1,1,1]$} & \multicolumn{4}{c}{$\mathbf{\delta} = [1,1,1,1]$} \\\cline{1-9}
			Candidate head   & User 1     & User 2     & User 3     & User 4     & User 1     & User 2     & User 3     & User 4 \\ \hline
			$u_{1}$ & $3.9155$   & $3.7973$   & $3.7973$   & $3.7973$   & $3.9029$   & $3.7935$   & $3.7935$   & $3.7935$ \\
			$u_{2}$ & $3.7954$   & $3.9122$   & $3.7954$   & $3.7954$   & $3.7950$   & $3.9121$   & $3.7951$   & $3.7951$ \\
			$u_{3}$ & $3.7948$   & $3.7948$   & $3.9103$   & $3.7948$   & $3.7944$   & $3.7946$   & $3.9102$   & $3.7946$ \\
			$u_{4}$ & $3.7948$   & $3.7948$   & $3.7948$   & $3.9103$   & $3.7944$   & $3.7946$   & $3.7946$   & $3.9102$ \\\hline
			GNP & \cellcolor{blue!25}$214.0044$ & $213.9364$ & $213.8605$ & $213.8605$ & $213.2454$ & \cellcolor{blue!25}$213.6849$ & $213.6091$ & $213.6091$ \\\hline 						     
		\end{tabular}
		GNP is short for generalized Nash product. Each colored value is the highest GNP among all four values and implies that the corresponding user is selected as the head.
	\end{threeparttable}
\end{table*}

\subsection{Setup}

We consider a set of $4$ nodes $\{1,2,3,4\}$ that are in proximity with each other\footnote{Though the number of users is few, the most fundamentals are revealed. Evaluation with more than four users has also been conducted, with the same observations.}. We assume that the communication technology they use is Wi-Fi Direct which supports star network topology \cite{wd}. Since it is built on traditional Wi-Fi infrastructure mode, Wi-Fi Direct can achieve typical Wi-Fi speeds. In a recent experimental study we conducted, it is found that the network is able to provide an average capacity of more than $4$ MB/s (equivalent to $32$ Mb/s) for local data dissemination \cite{mao2017performance}. We consider an average energy consumption of $2.85$ Joule/MB for both sending and receiving via Wi-Fi Direct \cite{iosifidis2017efficient}. At the beginning of the contact, each user has one data to share with others and the size of each data is $10$ MB (e.g., a short video clip or high-definition photo). 
The available airtime $T$ is $20$ seconds. 

\subsection{Energy Budget and Sensitivity to Energy Consumption}

Hereafter, we use budget and sensitivity to represent energy budget and sensitivity to energy consumption respectively. We set user $[2, 3, 4]$'s budget to $[500, 400, 400]$ Joules and vary user $1$'s budget in $\{50, 100, 300, 500\}$ Joules to see how head selection and airtime allocation are affected. We also consider two types of sensitivity, i.e., $\mathbf{\delta} = [1,1,1,1]$ where all the users have the highest sensitivity and $\mathbf{\delta} = [0,1,1,1]$ where user $1$ is insensitive because e.g., it has a power bank. In addition, the users have the same bargaining power and unit reward is set to $0.01$. 
\begin{figure}[t]
	\centering
	\subfigure{\includegraphics[width=0.4\textwidth]{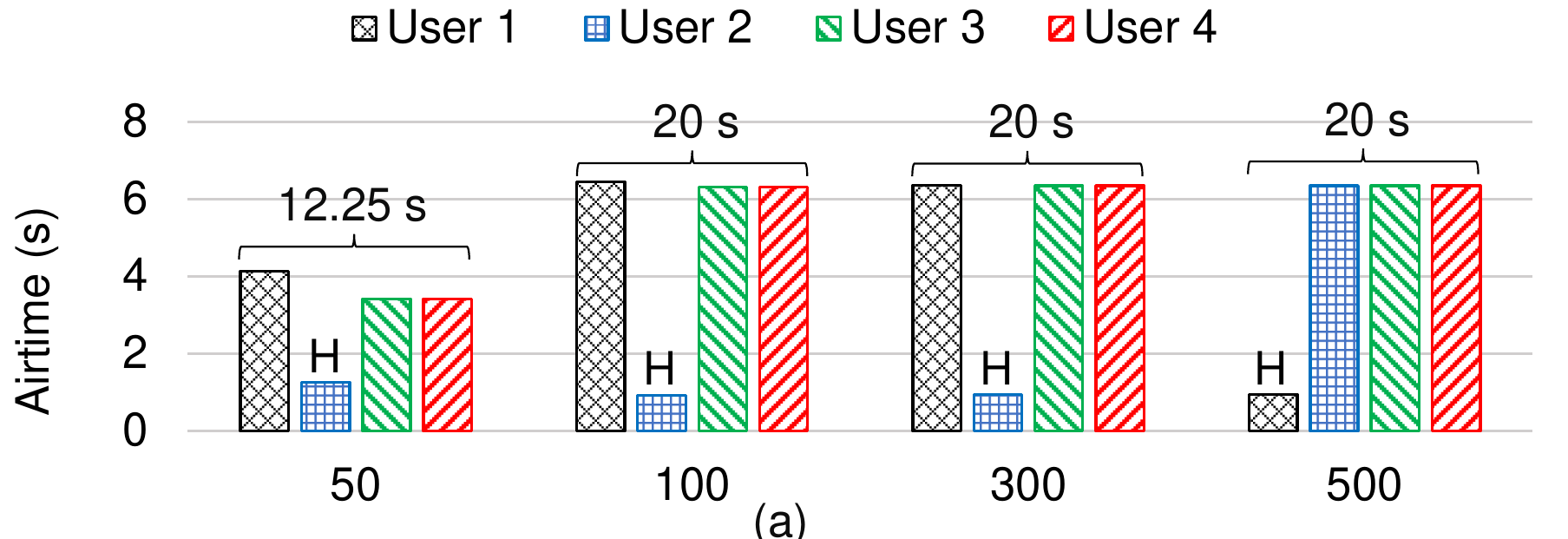}
		\label{fig-eb-airtime}}\\[-1ex]
	\subfigure{\includegraphics[width=0.4\textwidth]{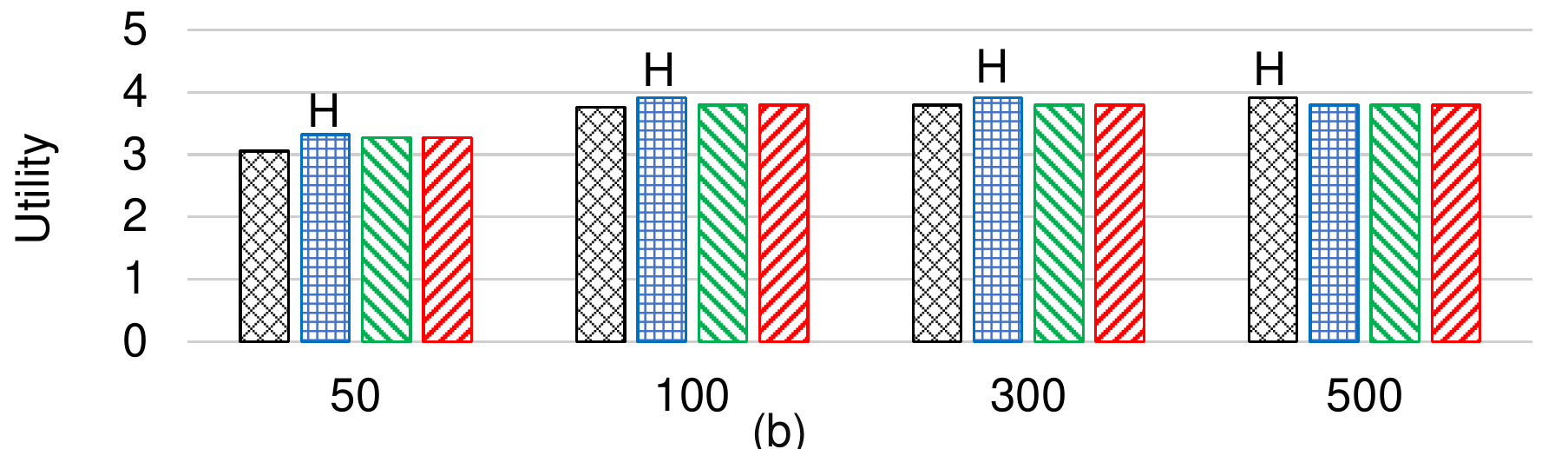}
		\label{fig-eb-utility}}\\[-1ex]
	\subfigure{\includegraphics[width=0.4\textwidth]{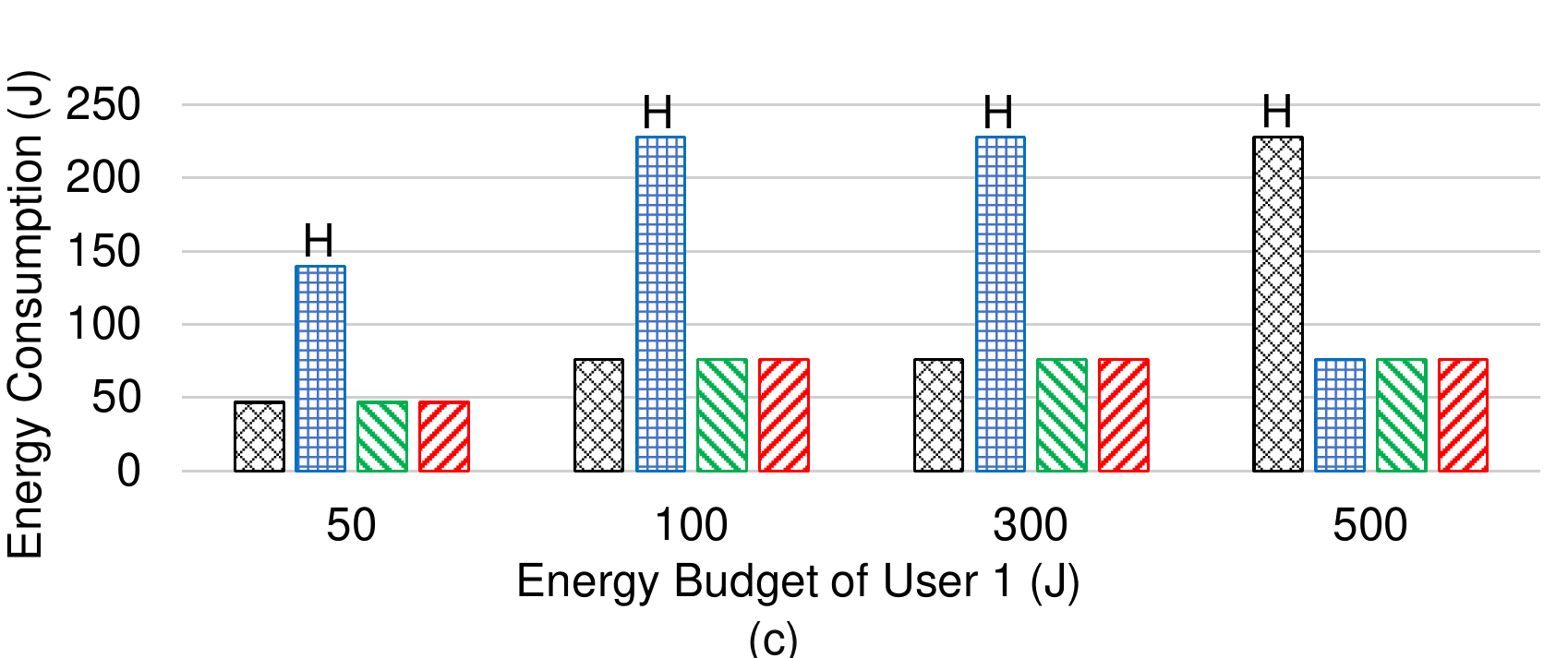}
		\label{fig-eb-energy}}\\[-1ex]
	\caption{Results of airtime allocation and head selection when $\mathbf{\delta} = [1,1,1,1]$.}
	\label{Fig-eb-sens1}
\end{figure}
Fig. \ref{fig-eb-airtime} shows the selected heads and airtime allocated to the dissemination of each user's data with varying budget of user $1$ and $\mathbf{\delta} = [1,1,1,1]$. In the figure, "H" above a bar is short for `head' and implies that the corresponding user is selected as the head. When the budget of user $1$ is $50$ Joules, only $12.25$ s out of $20$ s is allocated. The reason is that the NBS-based approach guarantees fairness in utility. As can be seen from Fig. \ref{fig-eb-utility}, the utilities of all the users are rather close except that the utility of the head is slightly larger than others (the reason will be explained in Section \ref{sec-sim-reward}). Fig. \ref{fig-eb-energy} shows that the budget of user $1$ will be totally utilized when it is $50$ Joules. Therefore, if all the $20$ s is allocated, only the utilities of user $[2, 3, 4]$ will increase while user $1$'s utility will not, which is not fair. However, this does not necessarily mean that the NBS-based approach is inefficient. On the contrary, it provides fairness in utility within each execution, the rest available time will be allocated among user $[2, 3, 4]$ in the next execution and fairness is guaranteed among them.

Fig. \ref{fig-eb-airtime} also shows that the user with the highest budget is always selected as the head, provided that their sensitivities 
are the same. This seems reasonable since the head consumes significantly more energy than others, as shown in Fig \ref{fig-eb-energy}. However, if the sensitivity of user $1$, i.e. $\delta_{1}$, is zero, user $1$ will become the head when its budget is $300$ Joules which is smaller than user $2$'s budget. This is because the cost $\delta_{1}(\frac{1}{E_{1}-e_{1}} - \frac{1}{E_{1}})$ in its utility function is always zero no matter how much energy is consumed. Table \ref{table-eb-ut} shows that, when user $1$ is the head, the utility of user $1$ and that of user $2$ is respectively larger than that of user $2$ and that of user $1$ when user $2$ is the head. As a result, user $1$ being the head has larger generalized Nash product than user $2$ being the head does. However, when user $1$'s budget is $100$ or $50$ Joules, it is not selected as the head simply because its budget does not support to utilize the whole $20$ s and therefore leads to low utility for all the users.

\subsection{Unit Reward}
\label{sec-sim-reward}

Fig. \ref{Fig-eb-sens1} and Table \ref{table-eb-ut} have shown that the selected head can have higher utility than other users. Especially, it holds no matter which user is finally selected as the head in the examples shown in Table \ref{table-eb-ut}. The form of utility function (\ref{utility}) implies that the forwarding reward to the head might be the cause. To study its impact on users' utility, we vary the unit forwarding reward $\gamma$ in $[0,0.02]$. In addition, we let all the users have the same budget $500$ Joules and the same sensitivity $1$. 

Fig. \ref{fig-rw-utility} shows, on one hand, that the utility of the head (user $1$) increases with the unit reward. However, when there is no reward or the unit reward is too small, i.e. in $[0,0.0012]$, the utility of the head is lower than that of other users as shown in the small window in Fig. \ref{fig-rw-utility}. In such cases, user $1$ may not be very willingly to be the head. When the unit reward is high enough, the head can gain a higher utility than other users. On the other hand, the utilities of the peripheral users (user $2, 3, 4$) also increase with the unit reward but slower than that of the head does and stop increasing when the unit reward is $0.013$ and higher. 
Fig. \ref{fig-rw-airtime} illustrates the airtime allocated to the users' disseminations.  It can be seen that higher unit reward motivates the head to allocate more airtime to other users until the `\textit{others-first}' point where all the airtime is allocated to the other users and it get zero airtime for the dissemination of its own data. From that point where $\gamma = 0.013$ onward, increasing unit reward will not change the airtime allocation anymore, and the peripheral users will not be able to disseminate or receive more data. This explains why their utilities keep unchanged when $\gamma = 0.013$ and higher in Fig. \ref{fig-rw-utility}. 
\begin{figure} [t]
	\centering
	\includegraphics[width=0.4\textwidth]{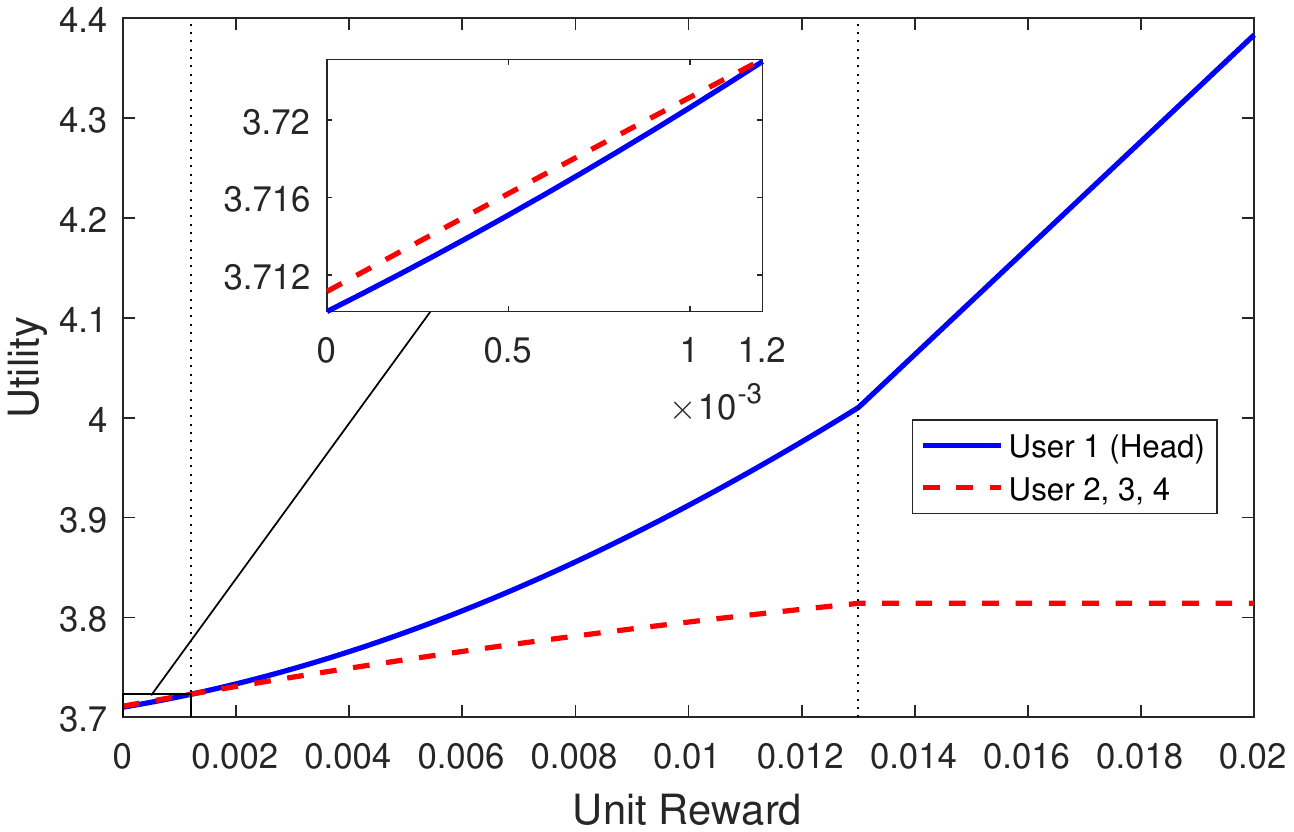}
	\caption{Utilities of users with varying unit reward. }
	\label{fig-rw-utility}
\end{figure}
\begin{figure} [t]
	\centering
	\includegraphics[width=0.4\textwidth]{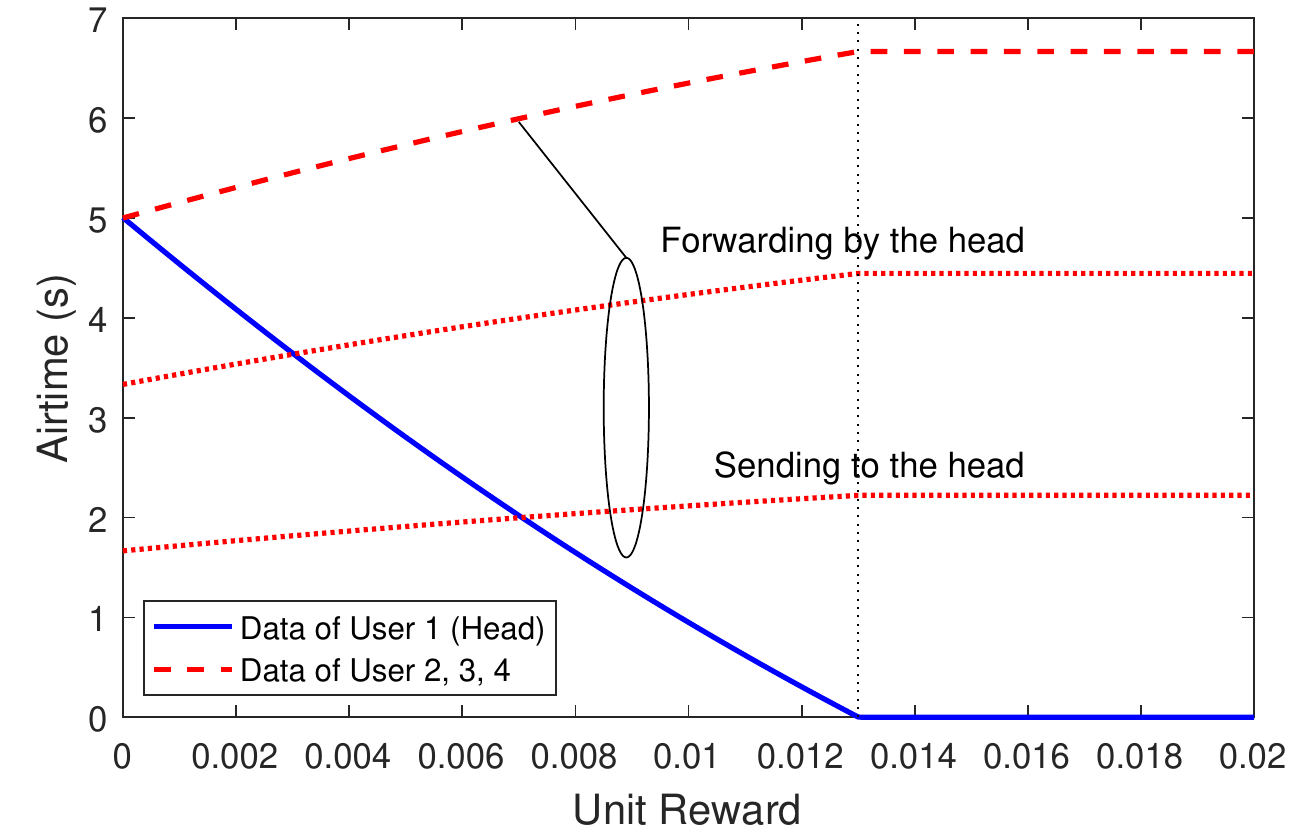}
	\caption{Airtime allocated to each user with varying unit reward. }
	\label{fig-rw-airtime}
\end{figure}

In summary, a higher unit reward certainly motivates the users to be the head. However, a higher unit reward does not necessarily motivate the head to forward more data for others due to the existence of the others-first point. Therefore, a unit reward lying between zero and the others-first point is recommended in real application, which can not only make a trade-off between disseminating the head's data and disseminating the peripheral users' data, but can also control the gap between  the users' utilities.

\subsection{Bargaining Power}
\label{sec-sim-bp}
It is expected that a user's utility increases with its bargaining power. Fig. \ref{fig-bp-utility} shows the utilities of the users with increasing bargaining power of the head. In a loose sense, the expectation is reasonable. However, with a higher unit reward, the head seems less keen on obtaining longer airtime for its own data (as can be seen from Fig. \ref{fig-bp-airtime}), since a higher unit reward already allows it to achieve much higher utility than others. From another point of view, changing the bargaining power is less effective than changing the unit reward to control the gap between the users' utilities. Especially when the unit reward is high (e.g., $\gamma = 0.02$), even the bargaining power of the head is much smaller than that of others, it still have much higher utility than others. Nevertheless, increasing the bargaining power of the head enables the head to obtain a higher utility than other users even when there is no reward for forwarding (i.e. $\gamma = 0$). 

\begin{figure} [t]
	\centering
	\includegraphics[width=0.4\textwidth]{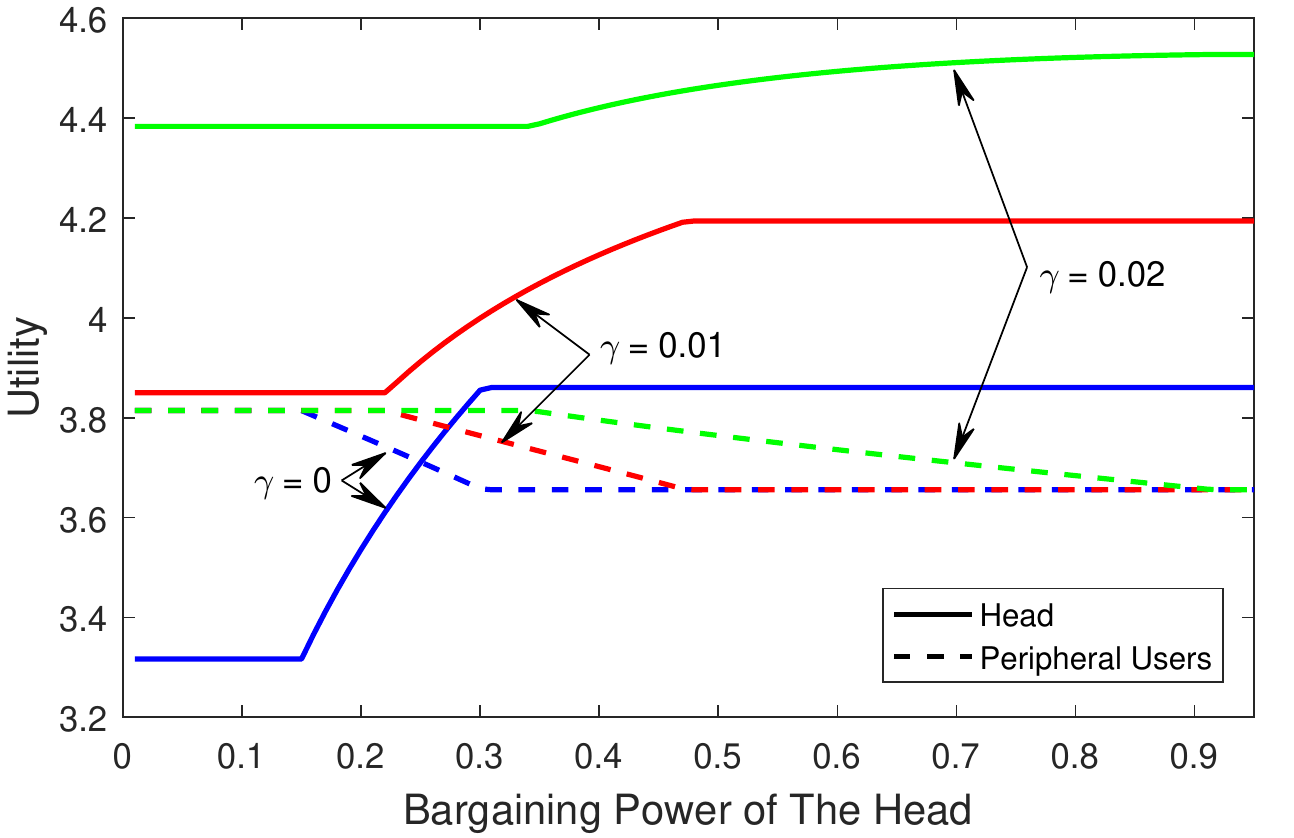}
	\caption{Utilities of users with varying bargaining power. }
	\label{fig-bp-utility}
\end{figure}

\begin{figure} [t]
	\centering
	\includegraphics[width=0.4\textwidth]{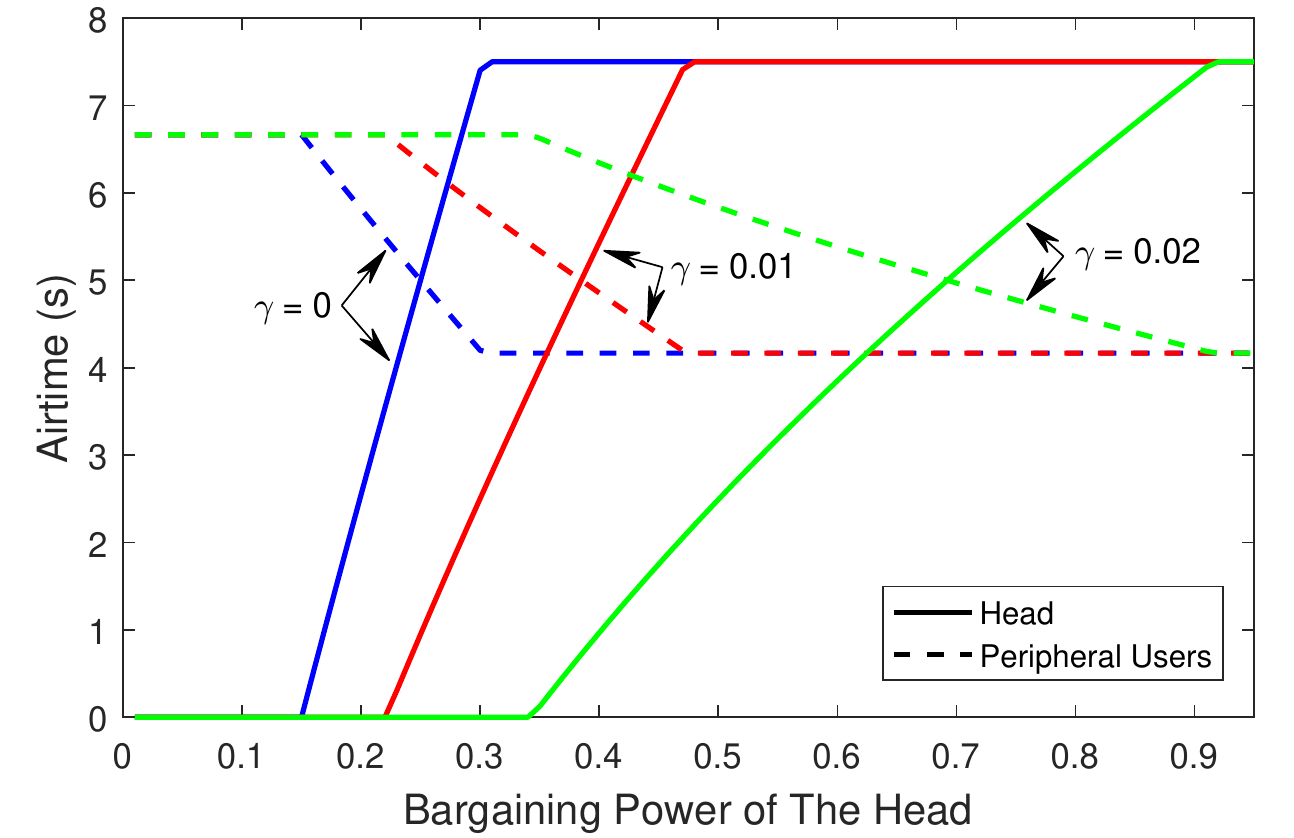}
	\caption{Airtime allocated to each user with varying bargaining power. }
	\label{fig-bp-airtime}
\end{figure}

\subsection{Data Load}
\label{sec-sim-DL}
Larger data load means that longer airtime is required to complete the dissemination. To investigate whether data load affects head selection and airtime allocation, we vary the data load of user $1$ from $2$ MB to $20$ MB and keep others' data load fixed to $10$ MB. Table \ref{table-dl-hd} shows the head selection under three different bargaining power settings. We can see that when the bargaining power of user $1$ is lower than or equal to that of others, the selected head does not change with user $1$'s data load. However, when the bargaining power of user $1$ is very high (i.e., $\frac{10}{13}$), the final head is changed from user $1$ to user $2$ when the data load of user $1$ becomes larger than that of others. To find the reason behind such a change, the airtime allocation with bargaining power $[\frac{10}{13},\frac{1}{13},\frac{1}{13},\frac{1}{13}]$ is plotted in Fig. \ref{fig-dl-airtime}. It can be seen that if user $1$ is the head, the airtime allocated to itself increases linearly with its data load. As a result, the gap between the utility of user $1$ and that of other users increases with user $1$'s data load, as can be seen in Fig. \ref{fig-dl-utility}. Larger utility gap gives smaller generalized Nash product and thus is regarded less fair by generalized NBS. That is why user $2$ is selected as the head by generalized NBS when user $1$ has a larger data load than others. 
\begin{table}[h]
	\caption{Head selection with different data load of user $1$.}
	\label{table-dl-hd}
	\centering
	\tabcolsep=0.19cm
	\begin{tabular}{c|cccccccccc} 
		\hline
		\multirow{2}{*}{Bargaining power}      & \multicolumn{10}{c}{Data load of user $1$ (in MB)} 
		\\\cline{2-11}
		& $2$   & $4$   & $6$   & $8$  & $10$  & $12$   & $14$   & $16$   & $18$  & $20$  \\ \hline
		$[\frac{1}{31},\frac{10}{31},\frac{10}{31},\frac{10}{31}]$         
		& $1$   & $1$   & $1$   & $1$  & $1$   & $1$    & $1$    & $1$    & $1$   & $1$   \\
		$[\frac{1}{4},\frac{1}{4},\frac{1}{4},\frac{1}{4}]$         
		& $1$   & $1$   & $1$   & $1$  & $1$   & $1$    & $1$    & $1$    & $1$   & $1$   \\
		$[\frac{10}{13},\frac{1}{13},\frac{1}{13},\frac{1}{13}]$         
		& $1$   & $1$   & $1$   & $1$  & $1$   & $2$    & $2$    & $2$    & $2$   & $2$   \\\hline
	\end{tabular}
\end{table}
\begin{figure} [H]
	\centering
	\includegraphics[width=0.4\textwidth]{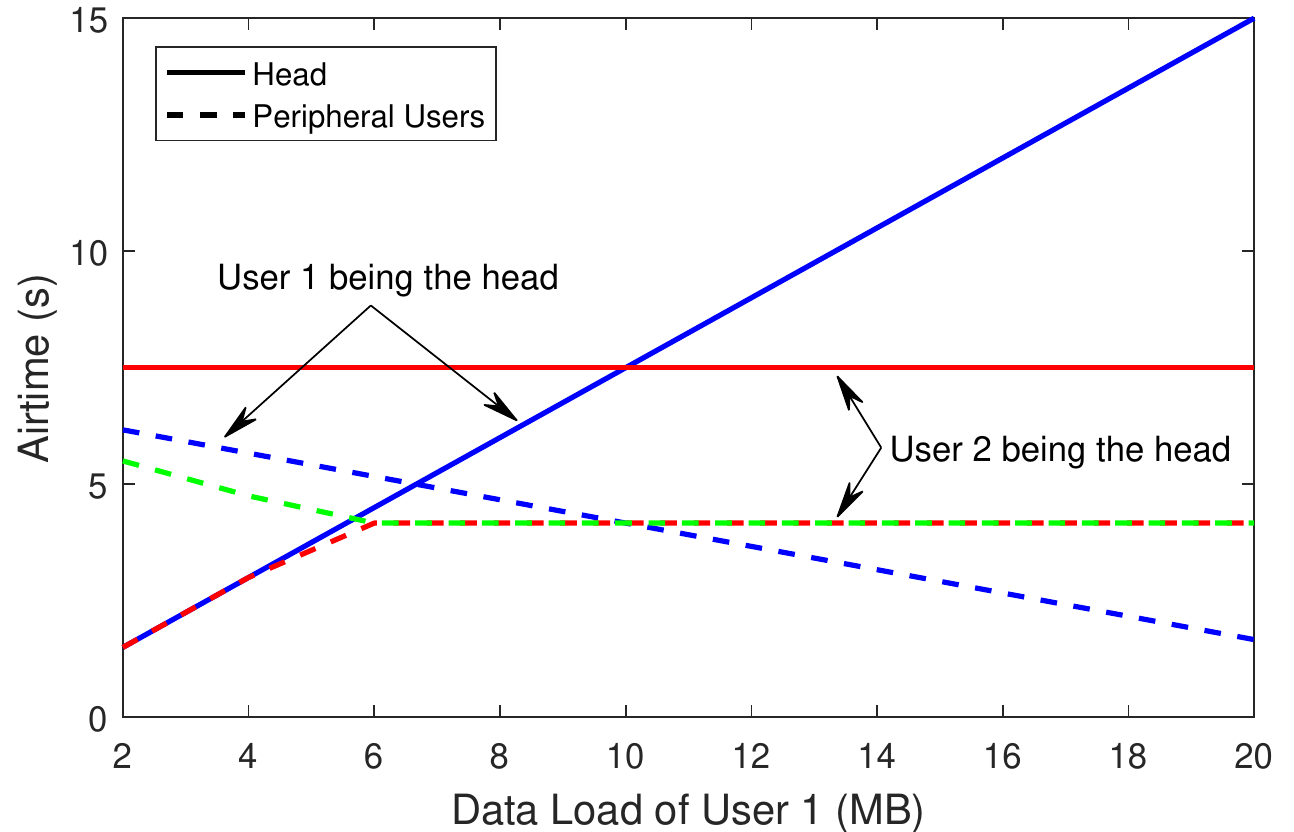}
	\caption{Airtime allocated to each user with varying data load of user 1. The green dotted line with user $2$ being the head shows the airtime allocated to user $3$ and $4$, while the red dotted line is the airtime to user $1$.}
	\label{fig-dl-airtime}
\end{figure}
\begin{figure} [H]
	\centering
	\includegraphics[width=0.4\textwidth]{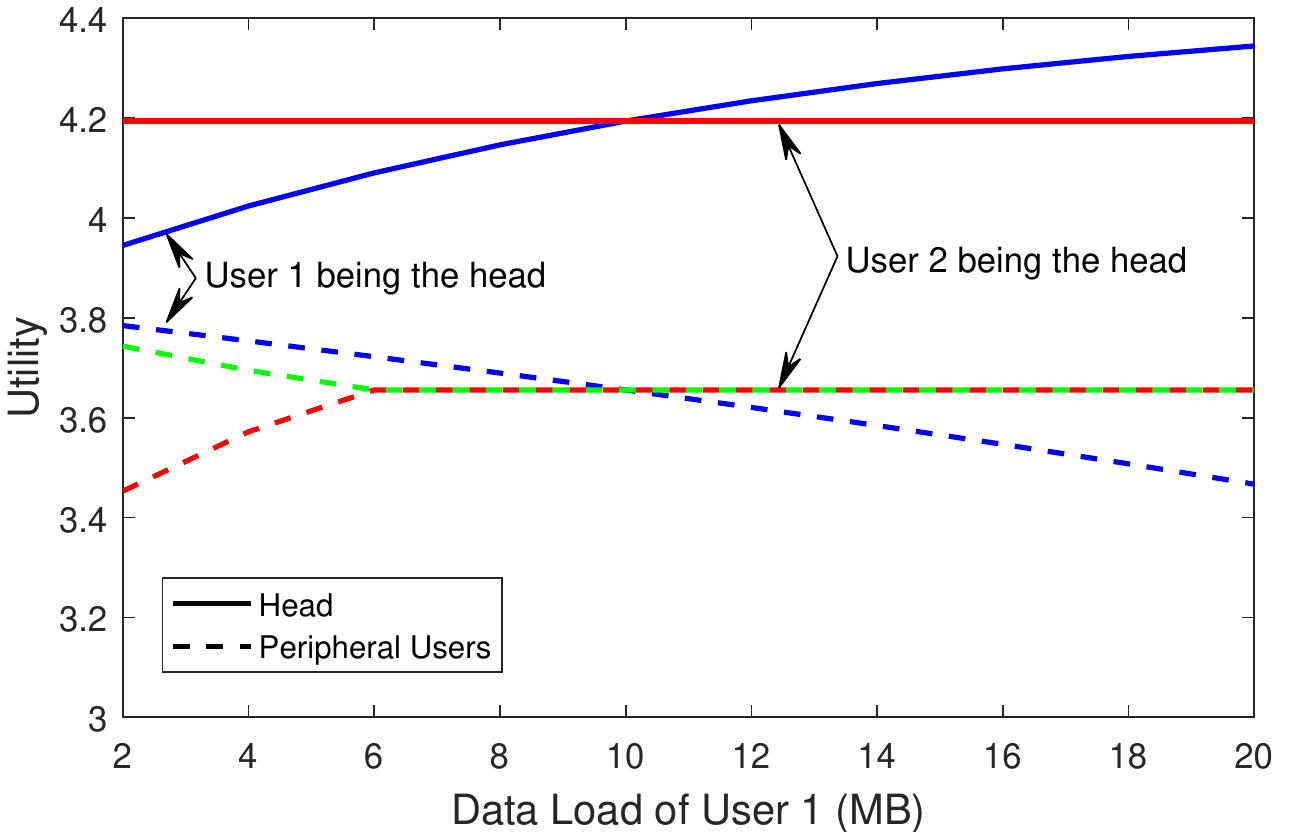}
	\caption{Utilities of the users with varying data load of user 1. The green dotted line with user $2$ being the head shows the utility of user $3$ and $4$, while the red dotted line is the utility of user $1$.}
	\label{fig-dl-utility}
\end{figure}

\subsection{Link Capacity and User Preference}
\label{sec-sim-UP}
In practice, links in the same network may have different capacities and users may have different preferences on the data. In the following, we consider that the links are symmetric (i.e., $c_{ij} = c_{ji}$) and their capacities are shown in Fig. \ref{fig-pr-lc}. Each of users $1$, $2$, and $3$ has one data (i.e., $A1$, $A2$, and $A3$ respectively) to share while user $4$ has two data $A(4,1)$ and $A(4,2)$ to share. We consider four cases of user preference: case 1 -- homogeneous preference, case 2 -- user $1$ is not interested in user $4$'s data $A(4,1)$, case 3 -- user $2$ is not interested in $A(4,1)$, and case 4 -- both user $1$ and $2$ are not interested in $A(4,1)$. 
\begin{figure} [ht]
	\centering
	\includegraphics[width=0.25\textwidth]{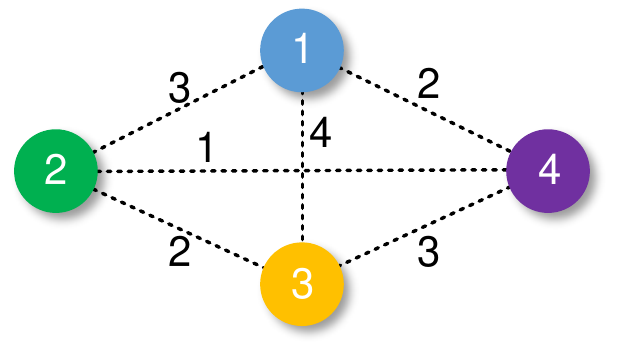}
	\caption{Link capacity (in MB/s).}
	\label{fig-pr-lc}
\end{figure}

\begin{table*}[tb]
	\caption{Head selection and the amount of data disseminated within $T$ for the four cases (In the table, GNP and TAoD are short for generalized Nash product and total amount of data disseminated within $T$ respectively.).}
	\label{table-pr-hd}
	\centering
	\begin{tabular}{c|cc|cc|cc|cc}
		\hline
		\multirow{2}{*}{Candidate head} & \multicolumn{2}{c|}{Case $1$} & \multicolumn{2}{c|}{Case $2$}& \multicolumn{2}{c|}{Case $3$} & \multicolumn{2}{c}{Case $4$}
		\\\cline{2-9}
		& GNP   & TAoD   & GNP   & TAoD  & GNP   & TAoD   & GNP   & TAoD   \\ \hline
		User $1$         
		& $138.7956$   & $55.3846$   & $138.7956$   & $55.3846$  & $138.7956$   & $55.3846$    & $138.7956$    & $55.3846$    \\
		User $2$         
		& $70.5070$   & $32.7273$   & $70.5070$   & $32.7273$  & $70.5070$   & $32.7273$    & $70.5070$    & $32.7273$    \\
		User $3$         
		& $138.7956$   & $55.3846$   & $138.7956$   & $55.3846$  & $150.4951$   & $59.2308$    & $140.7025$    & $56.1538$    \\
		User $4$         
		& $70.5070$   & $32.7273$   & $70.5070$   & $32.7273$  & $84.2542$   & $39.0909$    & $83.2800$    & $37.2727$    \\\hline
		Selected head        
		& User $1$   &    & User $1$   &   & User $3$   &     & User $3$    &     \\\hline
	\end{tabular} \\
\end{table*}

Table \ref{table-pr-hd} shows the results of head selection and the total amount of data disseminated for the four cases. Different from the case of identical link capacity, where the total amount of data disseminated does not change no matter which user is the head, we can see that the selected head for each case is the one that can disseminate the highest amount of data among all four candidate heads. Fig. \ref{Fig-pref-result} shows the results of airtime allocation for the four cases. Before giving the explanation, we define \textit{average dissemination rate} (ADR) of a given data $m$ of user $i$ to be the total amount of this data that has been disseminated to all interested users per second. Mathematically, it can be written by
\begin{equation}\label{avg-diss-rate}
ADR^{m}_{i} = \frac{N^{m}_{i}}{\sum_{(k,j)\in \mathcal{L}^{m}_{i}}\frac{1}{c_{kj}}}
\end{equation}
where $N^{m}_{i}$ is the number of users that are interested in the data. If the ADR of a data is higher, more bits of the data can be disseminated within the same airtime. It can be seen from Fig. \ref{fig-pref-airtime} that in case $1$, where the users have the same preference on all the data, the peripheral users are allocated the same airtime for their data's dissemination. Additionally, the airtime to user $4$ is shared by its two data equally. In comparison, user $4$'s data $A(4,1)$ in case $2$ does not get any airtime for its dissemination. The reason is that its ADR is relatively lower than that of any other data, as shown in Table \ref{table-pr-dissrate}. Therefore, allocating more airtime to other data rather than $A(4,1)$ would contribute more to all users' utilities. For the same reason, $A(4,1)$ in both case $3$ and $4$ is allocated long enough airtime so that it is totally disseminated, as can be seen in Fig. \ref{fig-pref-aod}. 
\begin{figure}[t]
	\centering
	\subfigure{\includegraphics[width=0.4\textwidth]{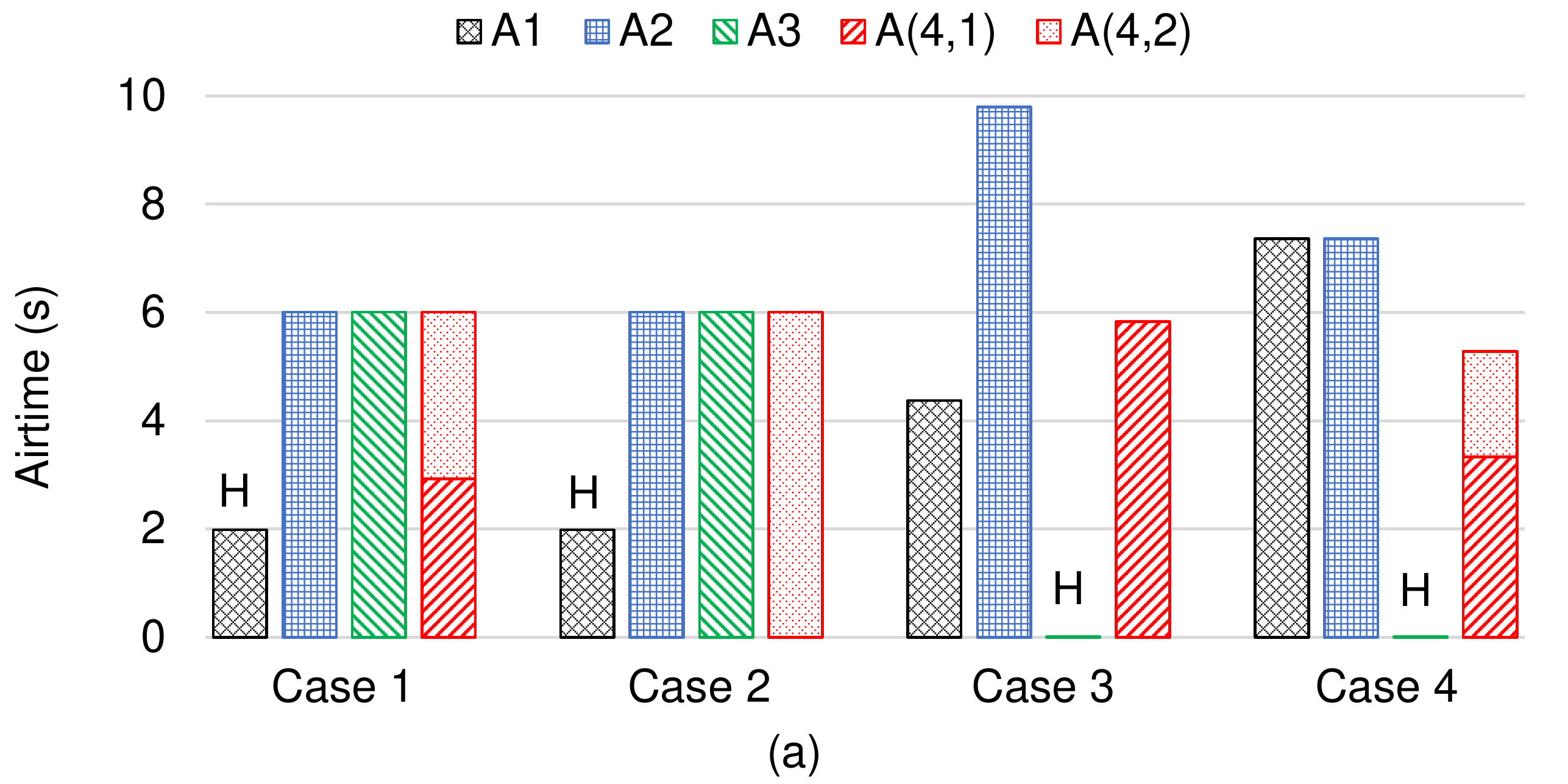}
		\label{fig-pref-airtime}}\\
	\subfigure{\includegraphics[width=0.4\textwidth]{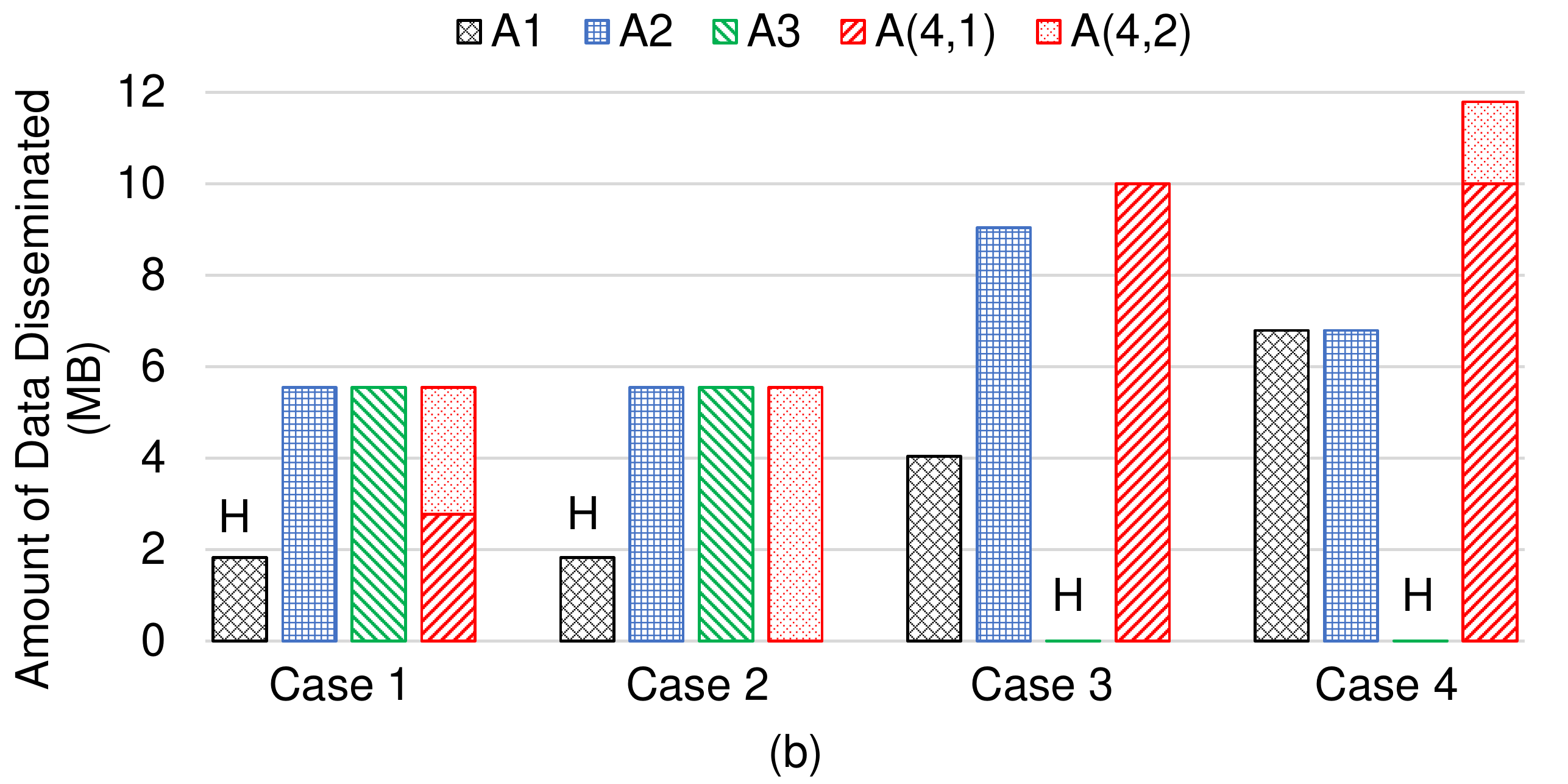}
		\label{fig-pref-aod}}\\[-1ex]
	\caption{Airtime allocation for different user preferences.}
	\label{Fig-pref-result}
\end{figure}

\begin{table}[t]
	\caption{The average dissemination rates of all the data for the four cases.}
	\label{table-pr-dissrate}
	\centering
	\begin{tabular}{c|cccc}
		\hline        
		Data & Case $1$   & Case $2$   & Case $3$   & Case $4$    \\
		\hline
		$A1$         
		& $2.7692$   & $2.7692$   & $2.7692$   & $2.7692$    \\
		$A2$        
		& $2.7692$   & $2.7692$   & $2.7692$   & $2.7692$    \\
		$A3$         
		& $2.7692$   & $2.7692$   & $2.7692$   & $2.7692$     \\
		$A(4,1)$
		& $2.7692$   & $1.8461$   & $3.4285$   & $3$     \\
		$A(4,2)$       
		& $2.7692$   & $2.7692$   & $2.7692$   & $2.7692$   \\\hline
	\end{tabular} \\
\end{table}

\subsection{Adaptive Head Selection and Airtime Allocation} 

The joint head selection and airtime allocation scheme presented in Section \ref{sec_nbs} assumes that the link capacities are constant during the contact. In reality, however, the link capacities may change with time-varying wireless channel condition and node mobility. In this section, we show how the scheme can be easily made adaptive to scenarios with time-varying link capacity and present numerical results for this adaptive scheme. 
\begin{table*}[t]
	\caption{Head selection for the ideal case with different slot sizes for the adaptive scheme.}
	\label{table-ideal-case}
	\centering
	\begin{tabular}{c|cccccccccccccccccccc|cccc} 
		\hline        
		\multirow{2}{*}{Slot size}   & \multicolumn{20}{c|}{\multirow{2}{*}{Head user in each slot of the duration ($20$)}} & \multicolumn{4}{c}{No. of times each user being the head} \\ 
		&&&& &&&& &&&& &&&& &&&&                                                                  & User $1$   & User $2$   & User $3$   & User $4$    \\
		\hline
		$20$      & \multicolumn{20}{c|}{\cellcolor{db}$1$} & $1$   & $-$   & $-$   & $-$ \\ \hline
		$10$      &\multicolumn{10}{c}{\cellcolor{db}$1$} &\multicolumn{10}{c|}{\cellcolor{tomato}$2$} & $1$   & $1$   & $-$   & $-$ \\ \hline
		$4$       &\multicolumn{4}{c}{\cellcolor{db}$1$} &\multicolumn{4}{c}{\cellcolor{tomato}$2$} &\multicolumn{4}{c}{\cellcolor{green}$4$} &\multicolumn{4}{c}{\cellcolor{tomato}$2$} &\multicolumn{4}{c|}{\cellcolor{yellow}$3$} & $1$   & $2$   & $1$   & $1$ \\ \hline
		$2$       &\multicolumn{2}{c}{\cellcolor{db}$1$} &\multicolumn{2}{c}{\cellcolor{tomato}$2$} &\multicolumn{2}{c}{\cellcolor{db}$1$} &\multicolumn{2}{c}{\cellcolor{yellow}$3$} &\multicolumn{2}{c}{\cellcolor{green}$4$} &\multicolumn{2}{c}{\cellcolor{tomato}$2$} &\multicolumn{2}{c}{\cellcolor{yellow}$3$} &\multicolumn{2}{c}{\cellcolor{green}$4$} &\multicolumn{2}{c}{\cellcolor{db}$1$} &\multicolumn{2}{c|}{\cellcolor{tomato}$2$}  & $3$   & $3$   & $2$   & $2$ \\ \hline
		$1$       &\cellcolor{db}$1$ &\cellcolor{tomato}$2$ &\cellcolor{green}$4$ &\cellcolor{tomato}$2$ &\cellcolor{yellow}$3$ &\cellcolor{db}$1$ &\cellcolor{yellow}$3$ &\cellcolor{db}$1$ &\cellcolor{green}$4$ &\cellcolor{tomato}$2$ &\cellcolor{green}$4$ &\cellcolor{yellow}$3$ &\cellcolor{db}$1$ &\cellcolor{tomato}$2$ &\cellcolor{yellow}$3$ &\cellcolor{tomato}$2$ &\cellcolor{green}$4$ &\cellcolor{db}$1$ &\cellcolor{yellow}$3$ &\cellcolor{green}$4$  & $5$   & $5$   & $5$   & $5$ \\
		\hline
	\end{tabular}
\end{table*}

\begin{figure*}[t]
	\centering
	\subfigure[Energy consumption]{\includegraphics[width=0.241\textwidth]{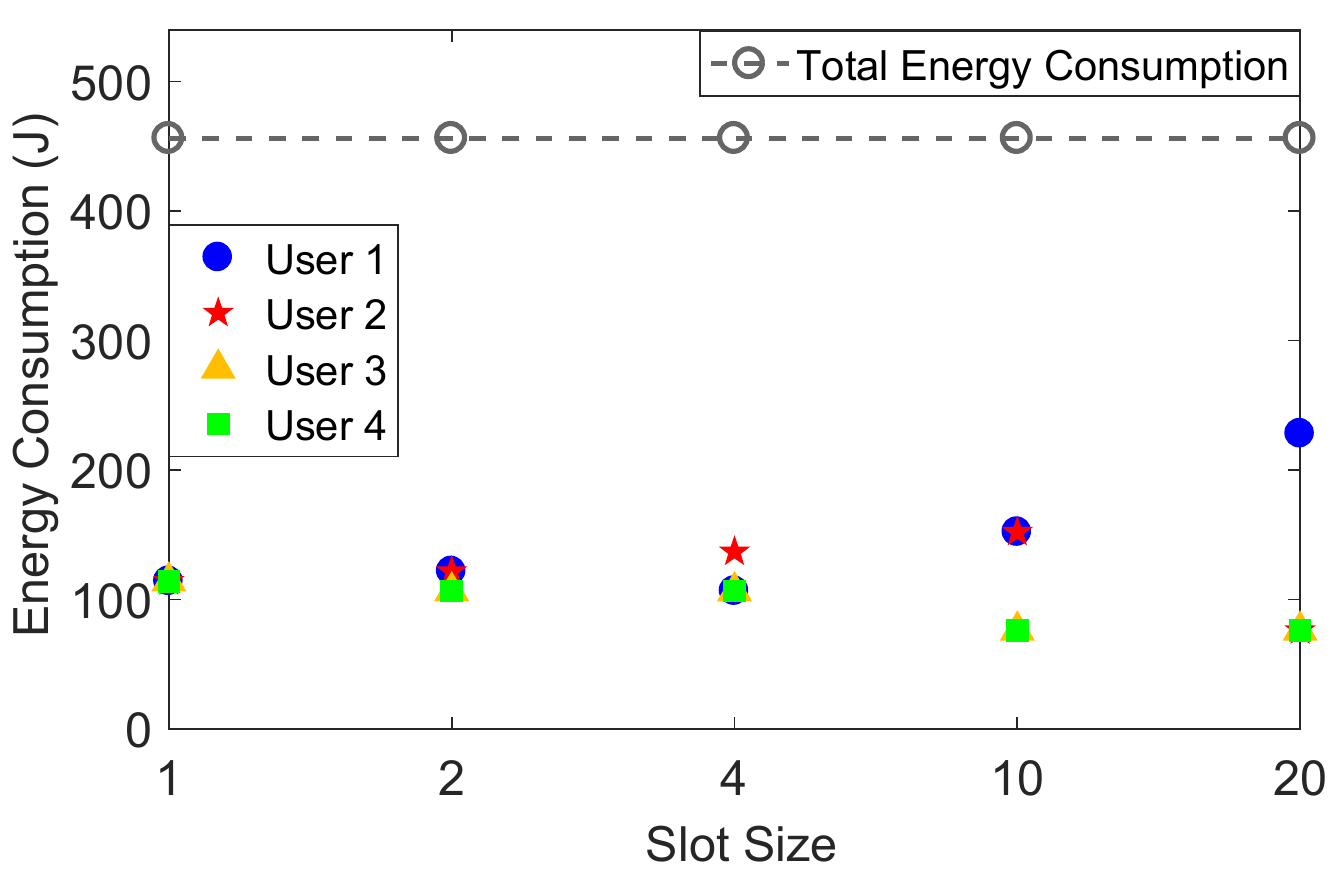} 
		\label{Fig:dy-same-energy}}
	\subfigure[Reward]{\includegraphics[width=0.241\textwidth]{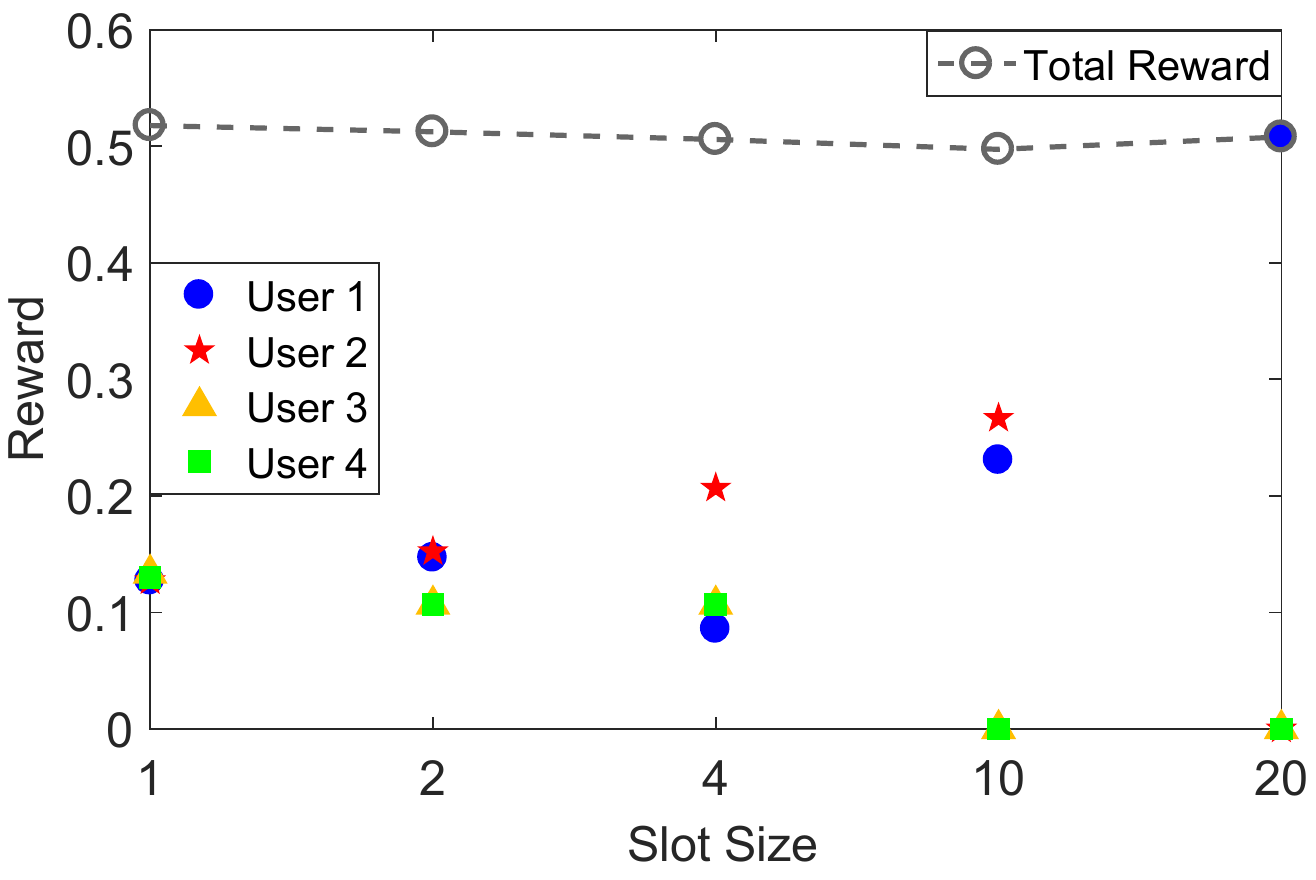}
		\label{Fig:dy-same-reward}}
	\subfigure[AoD]{\includegraphics[width=0.241\textwidth]{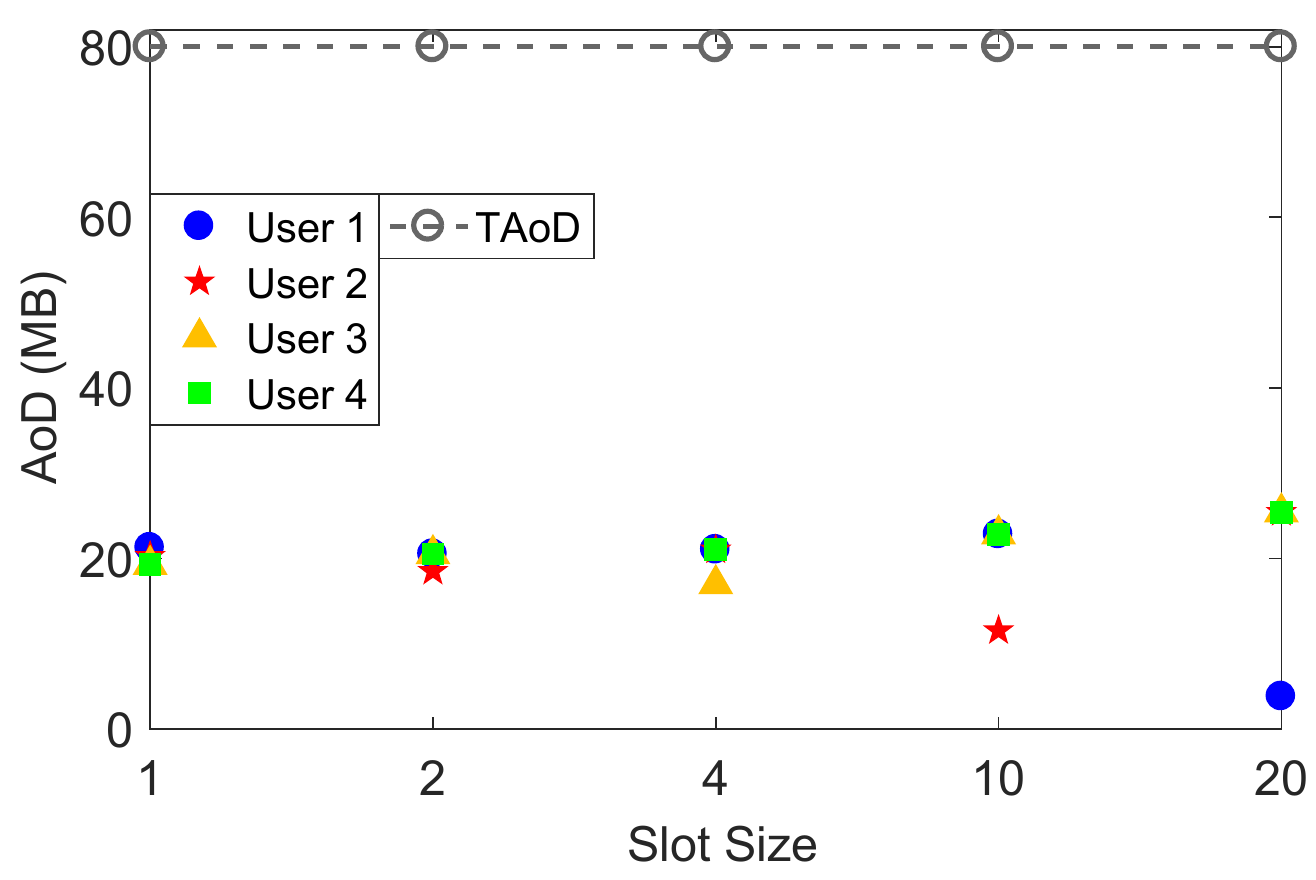}
		\label{Fig:dy-same-aod}}
	\subfigure[Utility and generalized Nash product]{\includegraphics[width=0.241\textwidth]{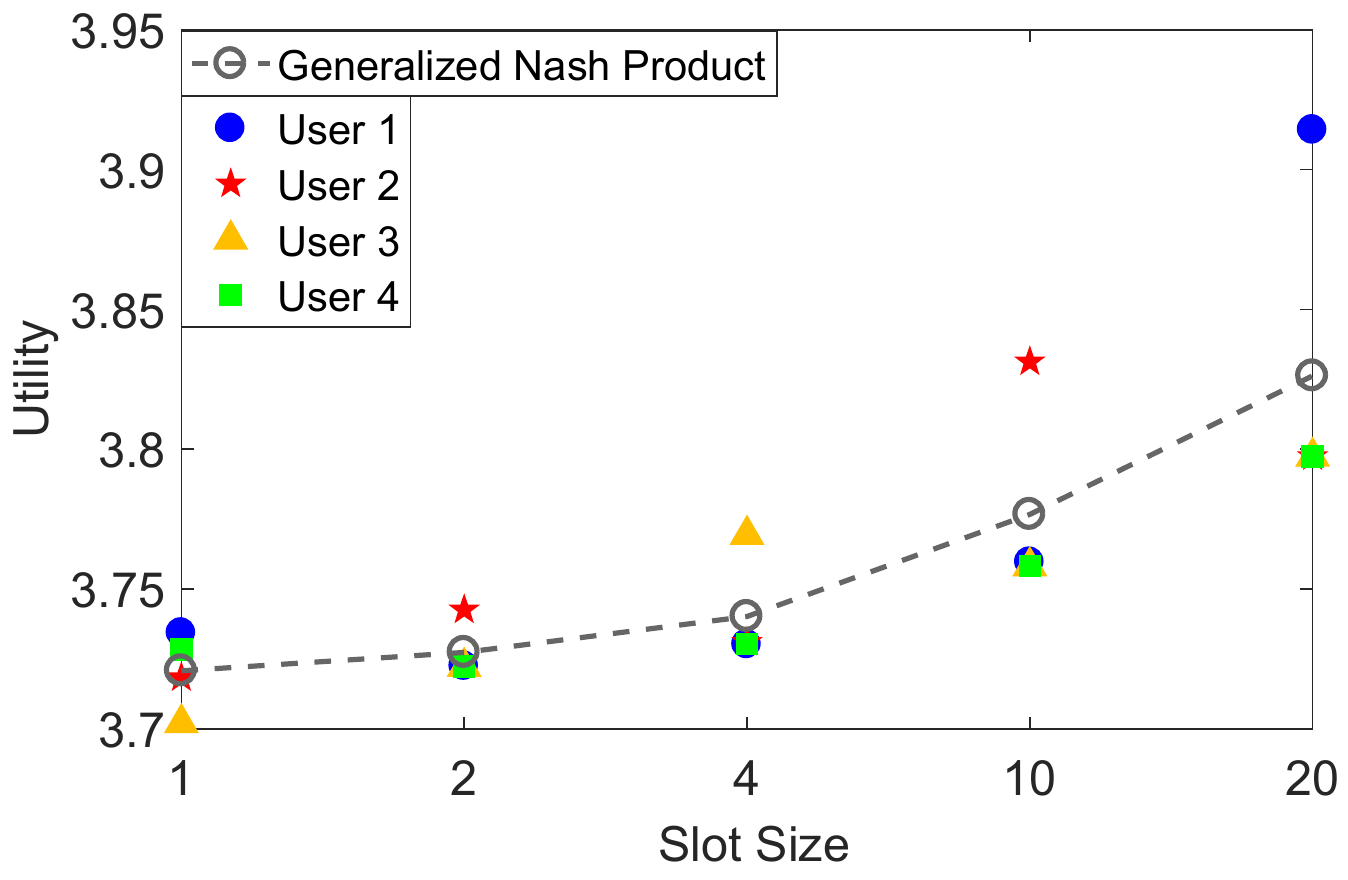}
		\label{Fig:dy-same-utility}}
	\caption{Results for the ideal case with different slot sizes for the adaptive scheme.}
	\label{Fig:ideal-case-result}
\end{figure*}

\begin{figure*}[t]
	\centering
	\subfigure[Energy consumption]{\includegraphics[width=0.2413\textwidth]{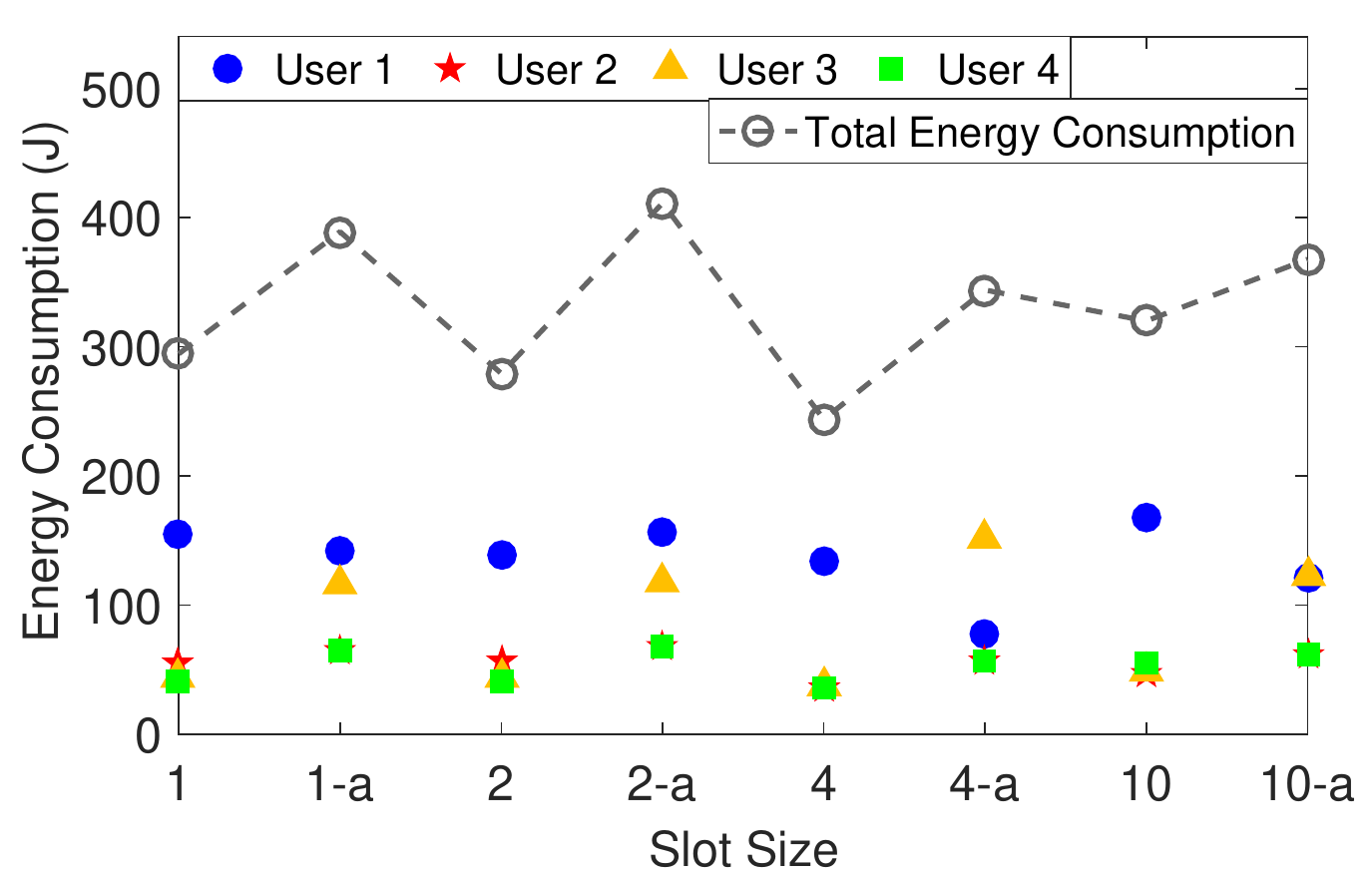} 
		\label{Fig:dy-energy}}
	\subfigure[Reward]{\includegraphics[width=0.2413\textwidth]{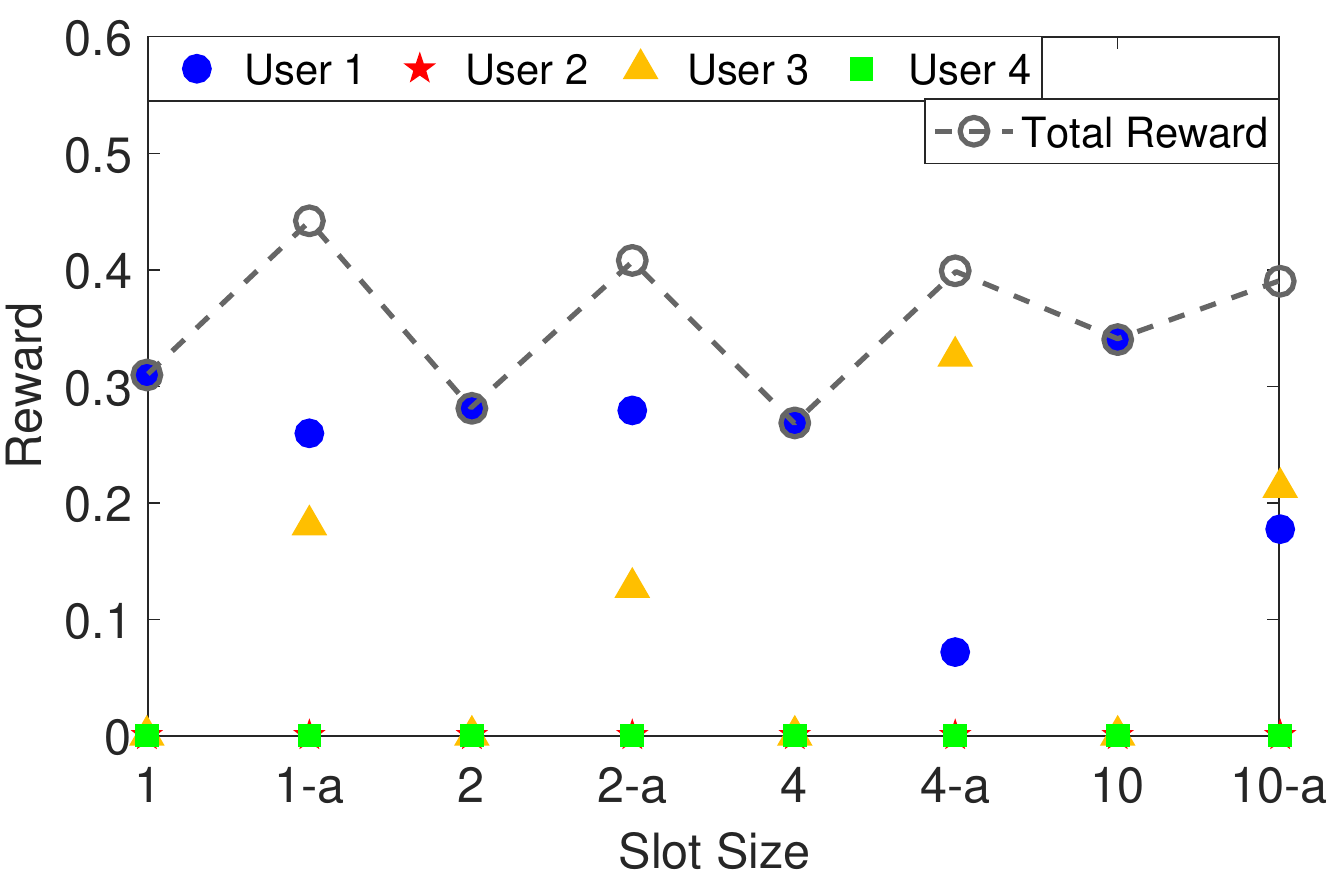}
		\label{Fig:dy-reward}}
	\subfigure[AoD]{\includegraphics[width=0.2413\textwidth]{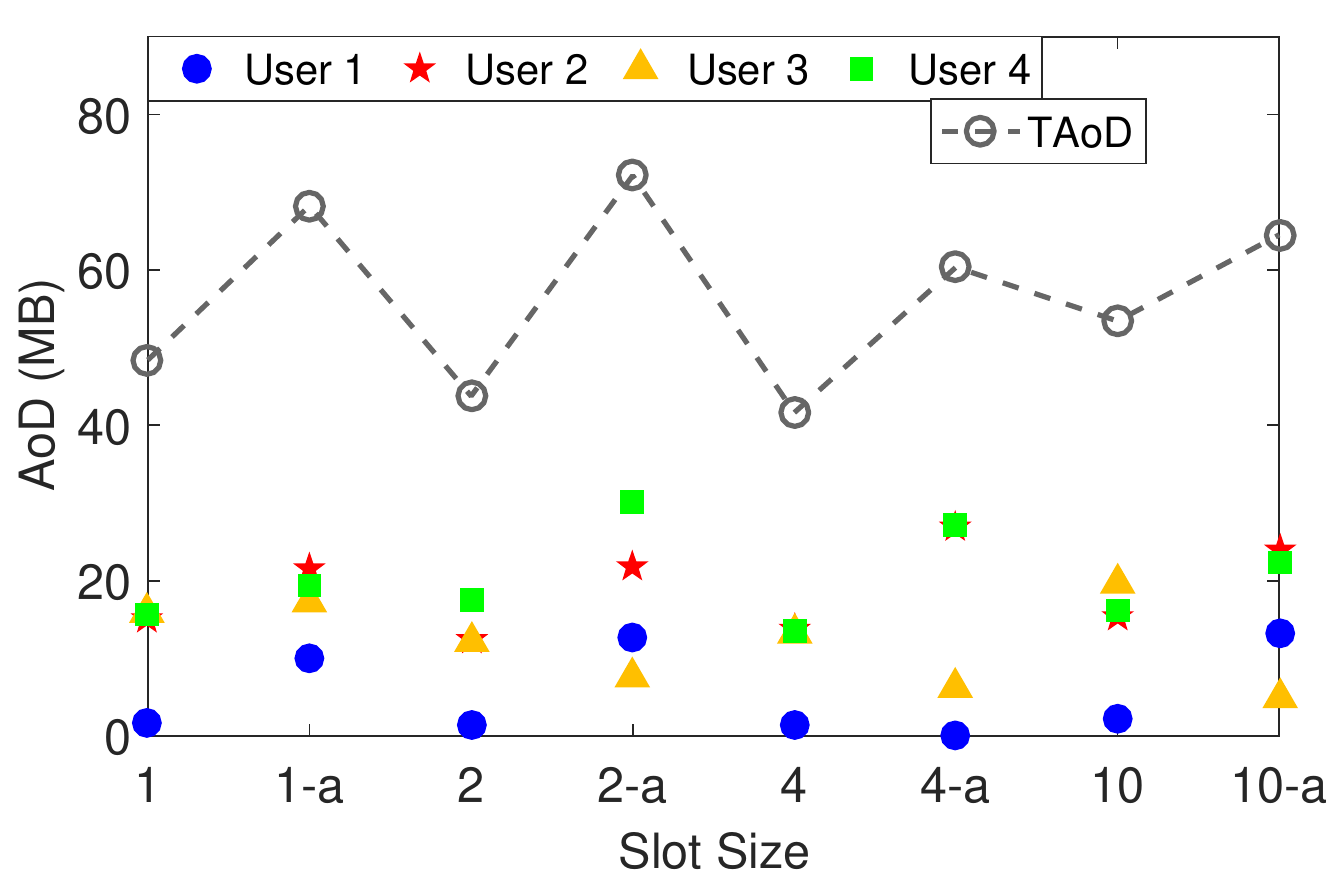}
		\label{Fig:dy-aod}}
	\subfigure[Utility and generalized Nash product]{\includegraphics[width=0.2413\textwidth]{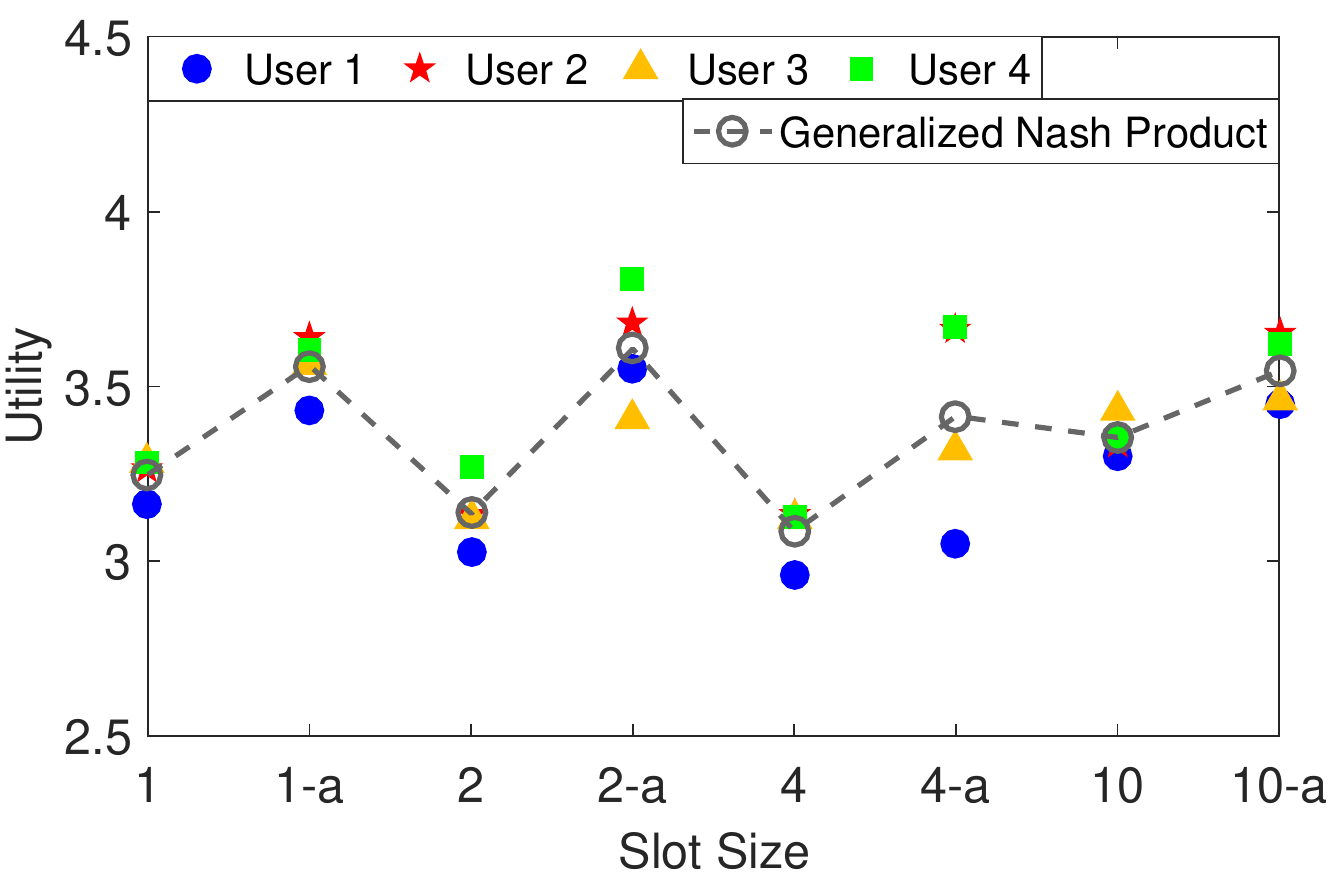}
		\label{Fig:dy-utility}}
	\caption{Comparison between the adaptive scheme and the non-adaptive scheme. Results for the adaptive scheme are marked with "-a" after slot sizes.}
	\label{Fig:rate-result}
\end{figure*}
The adaptive joint head selection and airtime allocation scheme is described as follows. Assume time is slotted with slot size $\pi$ and while the link capacity between two nodes may change over different slots, it remains constant within each slot. At the beginning of each slot $t$, the group of users perform a round of head selection and airtime allocation by solving the following \textit{slot-wise} optimization problem:
\begin{align}\label{full-pref-ad}
\underset{\mathbf{x,a}}{\max}\quad\quad & \textstyle\prod\limits_{i=1}^{N} u_{i}^{t}(\mathbf{x,a})^{\alpha_{i}} \\
\text{s.t.}\quad\quad             & \sum_{i=1}^{N}\sum_{m=1}^{M_{i}} x^{m}_{i} \leq \pi \label{c2-at-ad}\\
& 0 \leq x^{m}_{i} \leq \sum_{(k,j)\in \mathcal{L}^{m}_{i}}\frac{z^{m}_{i}}{c_{kj}}, & ~\forall i \in \mathcal{G}, m \in \mathcal{M}_{i} \\
& u_{i}(\mathbf{x,a}) \geq 0, &  ~\forall i \in \mathcal{G} \\
& e_{i} \leq E_{i}, &  ~\forall i \in \mathcal{G} \\
& a_{i} = \{1,0\}, &  ~\forall i \in \mathcal{G} \\
& \sum_{i=1}^{N} a_{i} = 1. \label{c2-head2-ad}
\end{align}
$u_{i}^{t}(\mathbf{x,a})$ is a slot-wise utility function for the users, which is expressed as
\begin{align} \label{utility-dy}
u_i^{t}(\mathbf{x,a}) = & v\Big(d_{i}^{t}(\mathbf{x})+\sum\limits_{k=1}^{t-1} d_{i}^{k} +b_{i}^{t}(\mathbf{x})+\sum\limits_{k=1}^{t-1} b_{i}^{k}\Big) \\
 - & g\Big(e_{i}^{t}(\mathbf{x,a})+\sum\limits_{k=1}^{t-1} e_{i}^{k}\Big)+ a_{i}\gamma f_{i}^{t}(\mathbf{x})+\sum\limits_{k=1}^{t-1} rw_{i}^{k} \notag
\end{align}
where $d_{i}^{k}$, $b_{i}^{k}$, $e_{i}^{k}$, and $rw_{i}^{k}$ are the amount of data disseminated, the amount of data of interest received, the amount of energy consumed, and the reward gained by user $i$ in the $k$th slot, respectively. The amount of airtime to be allocated is the size of this slot and thus the summation of the airtime to be allocated is no greater than the slot size. 

To study the impact of the slot size on the results of head selection and airtime allocation, we consider an ideal case where energy budget, sensitivity, bargaining power, data load, link capacity and data preference are the same for all the users. From Table \ref{table-ideal-case}, we can see that the users take turns to be the head when the slot size is small. As a result, the energy consumption of the users with smaller slot size are more balanced than that of the users with larger slot size, which can be seen from Fig. \ref{Fig:dy-same-energy}. This result coincides with the idea of periodic clusterhead rotation \cite{heinzelman2000energy,younis2004heed} in the clustering schemes for energy-constrained networks such as wireless sensor networks (WSNs) in order to balance energy usage. In addition to energy consumption, the amount of reward and the amount of data disseminated (AoD) of the users are also quite even when the slot size is small. Though small slot size can balance specific costs and gains such as energy consumption and reward, from Fig. \ref{Fig:dy-same-utility}, it seems that frequent slotting also results in low user utilities and low generalized Nash product.

In the following, we consider that the link capacities may vary over slots. We compare the adaptive scheme and the non-adaptive scheme that selects a head and allocates airtime only at the beginning of the contact regardless of link capacity changes. For each link, we assume it is a Rayleigh fading channel. For such a channel, its capacity $c$ has the following probability density function (PDF) \cite{Hogstad2009Exact}
\begin{equation}\label{rayleigh-capacity-pdf}
p(c) =  \frac{\ln(2)}{\rho}\cdot 2^{c}\cdot \mathrm{e}^{- (2^{c} - 1)/\rho}, \quad c \geq 0
\end{equation}
where $\rho$ represents the signal-to-noise ratio (SNR). 

Fig. \ref{Fig:dy-energy} shows that all users together consume more energy by using the adaptive scheme than using the non-adaptive scheme\footnote{To make the comparison reasonable, we have used a fixed value for $\rho$ here.}. It implies that the adaptive scheme better utilizes the limited airtime by adapting the head selection and airtime allocation to link capacity changes: it allows users to disseminate more data (Fig. \ref{Fig:dy-aod}) and gain more reward (Fig. \ref{Fig:dy-reward}) in total within the contact. More importantly, Fig. \ref{Fig:dy-utility} shows that users can obtain higher utilities in general and thus higher generalized Nash product.

\section{Related Work}
\label{sec_rw} 

Existing protocols on data dissemination in MSNs (e.g. \cite{ioannidis2009optimal,hui2011bubble,wang2013cloud,lu2016towards}) mostly focus on selecting proper data carriers, i.e. whether users should exchange data when they meet, while how data is exchanged after they decide to exchange data is often neglected or simplified. The major reason is that previous studies on MSNs predominantly assume that nodes contact with each other in a pairwise manner. And this assumption makes problem arising in data exchange, such as airtime allocation, seemingly trivial and therefore overlooked by previous studies. However, simultaneous multiple contact is quite common in many cases such as conference, underground, and tourist sites. This viewpoint is supported by a recent study on real-world contact traces \cite{wennerstrom2015considering}. Clearly, group communication among multiple contacting nodes can be more efficient than pairwise communication for content dissemination if multiple users are in contact. And the problem of airtime allocation for content dissemination within a group is nontrivial.

Since the network setting is basically the same, our work is quite related to airtime (or rate) allocation in WLANs and cellular access networks \cite{tan2004time,jiang2005proportional,ramjee2006generalized,li2008proportional}. However, head selection is out of the picture in these studies, because the head, i.e. access point in WLANs or base station in cellular networks, is provided by the service provider as part of the infrastructure. As a result, airtime allocation in these studies is only among users connecting to the AP or BS, while in our work, the head, like other users connecting to it, also competes for a share of the airtime. Another major difference between our work and airtime allocation in WLANs, cellular networks as well as mobile ad hoc networks (MANETs) is that we incorporate user's different preference on data in the user utility. Disseminating a piece of data does not only contribute to the utility of the data disseminator but also contribute to the utilities of all interested receivers. This social property has not been considered in previous studies on airtime allocation in WLANs, cellular networks or mobile ad hoc networks (MANETs). In our earlier work \cite{mao2017contact}, we propose a fair airtime allocation scheme for group data dissemination based on NBS. However, this scheme assumes the group head is already selected, and does not characterize user specifics including data dissemination need, preference on data, energy cost, etc. In this study, we jointly address the problems of head selection and airtime allocation, and incorporate the user preference on data, energy consumption, and forwarding reward into the NBS models.
 
In the literature of multihop wireless networks such as WSNs \cite{abbasi2007survey} and MANETs \cite{yu2005survey}, head selection is often addressed together with cluster formation. The purpose of such joint consideration is to control network topology so that the network nodes can joint force to achieve some specially targeted common objectives such as environment monitoring in an optimized manner (e.g. optimal energy efficiency and coverage). In all such studies, an underlying assumption is that nodes will follow such network-wide coordination to decide their roles and contribute. However, in MSNs, nodes are not intentionally deployed and they can have different interests. In addition, in MSNs, users move freely and decide if to join and what role (as group head or peripheral) to play in a group on their own will. For this reason, we have left cluster formation out of the scope of our work. Nevertheless, since for each cluster or group, head selection and airtime allocation are still necessary processes, our proposed scheme may be adopted after cluster formation has been performed.

\section{Conclusion and Future Work}
\label{sec_con}

In this paper, for local resource management in WSNs, we investigated airtime allocation among users of a group. Taking into consideration the unique characteristic of MSNs that the potential head also has to be counted in the allocation, the decision of airtime allocation has been performed jointly with head selection. To model this joint problem, a game-theoretic approach was proposed, and a Nash bargaining problem was formulated. We proved that, for both considered cases, i.e. a homogeneous case and a more general heterogeneous case, the NBS problem has at least one optimal solution, which ensures an acceptable outcome by all participating users with several properties of the NBS solution, including Pareto optimal and proportionally fair. In addition, we introduced a decomposition-based algorithm to find a unique optimal solution to the allocation problem, which allows being implemented in a distributed fashion. Through the numerical results, we found that the outcome of the head selection and airtime allocation scheme is affected by many parameters. Notably, the scheme prefers user with high energy budget, low sensitivity to energy consumption, and high dissemination rate to be the head. In addition, the reward for forwarding data for others can have significant impact on the allocation. 

We highlight that, in the modeling and analysis of the problem, some specific simple forms of the reward, utility function and energy consumption function have been used. In addition, in modeling the cost incurred to the corresponding user, we have only considered battery energy consumption for simplicity. Nevertheless, the essentials of the airtime allocation and head selection problem in MSNs are mostly revealed. For instance, the problem is so unique that existing optimal or fair airtime allocation schemes for use in classical WLANs and cellular access networks are not applicable or do not hold their initial design properties in MSNs. A future research direction is to consider other cost models for the head and other probably more general reward and utility functions in deciding the allocation. In addition, in this paper, we have considered a practical group topology, the star topology, where all communication has to go through the head that functions likes a personal WiFi hotspot. Another direction for future research is to study airtime allocation for MSNs based on mesh networking. In such a case, the problem is similar to that for ad hoc networks \cite{fang2004fair,lin2004joint,chen2006cross,tychogiorgos2012utility,vergados2018scheduling}, but more complex because of special MSN characteristics such as airtime limitation, incentive / reward involved, and heterogeneous user preference on data.

\ifCLASSOPTIONcaptionsoff
  \newpage
\fi




\bibliographystyle{IEEEtran}
\bibliography{msn-ref}

\end{document}